\documentclass[11pt]{article}
\usepackage[T1]{fontenc}
\usepackage{lmodern}
\usepackage[margin=1in]{geometry}
\usepackage{xcolor}
\usepackage{amsmath,amssymb,amsthm,mathtools}
\usepackage{microtype,booktabs,array}
\usepackage[colorlinks=true,linkcolor=blue!55!black,citecolor=blue!55!black,urlcolor=blue!55!black]{hyperref}
\hypersetup{pdftitle={Communication Complexity of Integer Matrix Singularity},
pdfauthor={Chengu Wang}}
\newtheorem{theorem}{Theorem}[section]
\newtheorem{proposition}[theorem]{Proposition}
\newtheorem{lemma}[theorem]{Lemma}
\newtheorem{corollary}[theorem]{Corollary}
\theoremstyle{definition}

\numberwithin{equation}{section}
\newcommand{\Q}{\mathbb Q}
\newcommand{\Z}{\mathbb Z}
\newcommand{\T}{\mathbb T}
\newcommand{\SINGplus}{\mathrm{SING}^{+}}
\newcommand{\SINGxor}{\mathrm{SING}^{\oplus}}
\newcommand{\SINGpart}{\mathrm{SING}^{\mathrm{part}}}
\newcommand{\bitand}{\mathbin{\&}}
\newcommand{\rank}{\operatorname{rank}}
\newcommand{\diag}{\operatorname{diag}}
\newcommand{\Unif}{\operatorname{Unif}}
\newcommand{\supp}{\operatorname{supp}}
\newcommand{\entry}{\mathrm{entry},1}
\newcommand{\corank}{\operatorname{corank}}

\title{Communication Complexity of Integer Matrix Singularity}
\author{Chengu Wang\\\small\href{mailto:wangchengu@gmail.com}{\texttt{wangchengu@gmail.com}}}
\date{}

\begin{document}
\hypersetup{pageanchor=false}
\pagenumbering{gobble}
\maketitle
\begin{abstract}
  Alice and Bob hold complementary halves of the columns of an
  $n\times n$ matrix of $k$-bit integers and decide whether
  it is singular over $\mathbb Q$. Up to constant-factor dimension growth,
  this partitioned problem is equivalent
  to additive singularity, where Alice and Bob hold $n\times n$ matrices $A,B$ of
  $k$-bit integers and test the singularity of $A+B$ over $\mathbb Q$.
  We prove a bounded-error quantum lower
  bound $\Omega(\min\{n^{3/2}\sqrt{k\log n},n^2\log(nk)\})$ and a
  randomized upper bound $O(n^2\min\{k,\log(k+\log n)\})$.
  The bounds match at $\Theta(n^2\log(nk))$ for $k\ge n\log n$.
  The lower bound holds even under the promise of rank $n-1$ versus $n$.
  Its proof combines an integer witness with a local digit encoding, while the
  upper-bound protocol lifts modular kernel relations using short messages.

  Bitwise sharing changes how the communication complexity depends on the entry bit length $k$.
  If Alice and Bob instead hold $n\times n$ matrices $A,B$ of $k$-bit
  integers and test whether their entrywise bitwise XOR $A\oplus B$
  is singular over $\mathbb Q$,
  we prove a quantum lower bound $\Omega(n^2k)$, matching the trivial upper bound.
  This lower bound holds even under the promise of rank $n-1$ versus $n$.
  Our proof uses a dense weighted
  equation and a signed Fourier witness whose exceptional dyadic
  frequencies can be removed at small variation cost.
\end{abstract}

\clearpage
\tableofcontents
\clearpage
\pagenumbering{arabic}
\hypersetup{pageanchor=true}

\section{Introduction}\label{sec:introduction}

Communication complexity, introduced by Yao~\cite{Yao79}, studies how much
two parties must communicate to compute a function of their separate
inputs~\cite{KN97}. Alice receives $x$ and
Bob receives $y$, and they exchange messages to compute $f(x,y)$.
Each message may depend on the sender's input and preceding messages;
local computation is unrestricted. Deterministic protocols must always
be correct. Randomized protocols may use shared random bits, while quantum
protocols exchange quantum states and perform local quantum operations.
Communication cost is the worst-case number of bits or qubits
transmitted. Bounded-error protocols must be correct with probability
at least $2/3$ on every promised input. We allow two-way communication
and, for quantum protocols, arbitrary entanglement shared before the
inputs arrive.
These models quantify how randomness and quantum resources reduce
communication.

Singularity asks whether a square matrix has a nonzero vector in its
kernel. A direct formulation
splits a matrix's columns between Alice and Bob. For even $n$,
they hold matrices $X,Y$ of size $n\times(n/2)$ with $k$-bit nonnegative integer
entries and decide whether $[\,X\ \ Y\,]$ is singular over $\mathbb Q$.
We denote this
balanced partition problem by $\SINGpart_{n,k}$. A balanced row
partition is equivalent by transposition.

We use an equivalent additive formulation: Alice and Bob hold $n\times n$
matrices $A,B$ of $k$-bit nonnegative integers and decide whether $A+B$ is singular.
We denote this problem by $\SINGplus_{n,k}$. Lemma~\ref{lem:addition-column-partition}
gives local reductions in both directions, preserving the entry
alphabet and corank with constant-factor dimension growth.
Thus the additive bounds below also hold for balanced
partitions, up to a constant factor. We work in the additive model.

We then ask how communication changes when these $k$-bit words
are combined bitwise. In the bitwise-XOR problem
$\SINGxor_{n,k}$, Alice and Bob hold $n\times n$ matrices $A,B$
and decide whether $A\oplus B$ is singular over $\mathbb Q$.
This changes the representation of the distributed matrix while
retaining the rational-singularity predicate.

We call these problems \emph{integer matrix singularity}, following the
input terminology of Chu and Schnitger~\cite{chu-schnitger}. The field
matters: an integer matrix is singular over
$\mathbb Q$ exactly when its integer determinant is zero, equivalently
when it has a nonzero integer kernel vector. The integer and
rational formulations therefore describe the same problem. They
do not describe singularity over a finite field: reduction modulo a
prime may turn a nonzero integer determinant into zero.

The bit length $k$ matters even though the output is
one bit. Sending one input uses $n^2k$ bits, and deterministic protocols
cannot improve this order. Randomization changes the picture for
addition: a small modular image can certify a nonzero determinant,
and the dependence on $k$ becomes logarithmic. Bitwise XOR has a
different behavior. Its bounded-error complexity remains proportional
to the complete input length, even with entangled quantum communication.
Our results quantify this contrast and identify
the remaining gap for additive inputs of small bit length.

\subsection{Previous results}

Chu and Schnitger's bounds~\cite{chu-schnitger}
imply $D(\SINGplus_{n,k})=D(\SINGxor_{n,k})=\Theta(n^2k)$,
even under the promise of rank $n-1$ versus $n$
(Lemma~\ref{lem:deterministic-singularity}).
Here $D$ denotes deterministic communication complexity.

The same paper records a randomized upper bound
$O(n^2\max\{\log n,\log k\})$ for partitioned integer singularity,
crediting Leighton (personal communication, 1987)~\cite[p.~396]{chu-schnitger}.
Its random-prime idea extends to additive inputs. Alice sends $A$
modulo a public prime, and Bob tests $A+B$ over the corresponding
field. Hadamard's inequality bounds the magnitude of a nonzero
determinant and hence the number of primes that can divide it.
Together with the trivial upper bound $O(n^2k)$, this gives the baseline
$O(n^2\min\{k,\log(nk)\})$.

Streaming linear algebra supplies a different lower bound.
Clarkson and Woodruff~\cite{clarkson-woodruff}
study rank and other linear-algebraic tasks on streamed matrices.
Their Theorem~5.2 and the following invertibility
consequence give an $\Omega(n^2)$ lower bound for one-pass streaming
and the corresponding one-way communication task. This one-way bound
does not imply an unrestricted interactive lower bound.

Finite-field singularity also has strong communication lower bounds.
Sun and Wang~\cite[Theorem~3]{sun-wang} prove randomized and quantum
lower bounds of order $n^2\log p$ for additive singularity over
$\mathbb F_p$. Li, Sun, Wang, and Woodruff~\cite[Theorems~1--2]{li-sun-wang-woodruff}
obtain tight bounds for linear-algebraic problems in the message-passing
model, including the two-player distinction between rank $n-1$ and
rank $n$. Sherstov and Storozhenko~\cite[Theorem~I.1]{sherstov-storozhenko} develop
lower bounds for approximating finite-field matrix rank, including
strong promises and small success advantages. These works demonstrate
that spectral and Fourier techniques can control quantum protocols
for distributed rank. These bounds depend on the specified field and
predicate.

For rational-XOR singularity, earlier communication lower-bound tools
already imply two baselines. Packing the input bits into words
and using a determinant encoding transfers Razborov's disjointness
lower bound~\cite{razborov} to a randomized $\Omega(n^2k)$ lower bound.
For Boolean inputs, Zhang and Shi~\cite[Theorem~2]{zhang-shi-xor} give
an $\Omega(N)$ entangled quantum lower bound for testing an exact
Hamming distance in the interior of the range. A local Boolean matrix
encoding of order $O(\sqrt N)$ turns this into an $\Omega(n^2)$ lower
bound for rational-XOR singularity. The determinant representation of
a sum used here also appears in Ivanyos et al.\ \cite[Claim~2 in the
proof of Theorem~4]{ivanyos-et-al}. Sections~\ref{sec:randomized-xor-lower-bound}
and~\ref{sec:boolean-xor-lower-bound} spell out the reductions. We treat these
classical and Boolean bounds as elementary consequences of the cited
theorems, rather than as new lower-bound contributions.

Vempala, Wang, and Woodruff~\cite{vempala-wang-woodruff} study
distributed linear systems and optimization with bounded integer
coefficients. Their random-prime feasibility protocol (Section~4.1)
uses the same modular idea as our baseline upper bound. Their lower
bounds cover both feasibility testing and solution output; comparisons
with singularity must account for the predicate, input distribution,
and error model.

Table~\ref{tab:prior-work} summarizes the closest comparisons.

\begin{table}[ht]
  \centering
  \small
  \begin{tabular}{@{}>{\raggedright\arraybackslash}p{.2\linewidth}>{\raggedright\arraybackslash}p{.25\linewidth}>{\raggedright\arraybackslash}p{.45\linewidth}@{}}
    \toprule
    Work & Setting & Relevant conclusion\\
    \midrule
    Chu--Schnitger~\cite{chu-schnitger} & Partitioned integer matrix &
    Deterministic $\Theta(n^2k)$ for multi-bit inputs; random-prime
    upper bound $O(n^2\log(nk))$.\\
    Clarkson--Woodruff~\cite{clarkson-woodruff} & One pass / one way &
    Quadratic lower bound for rank decision and invertibility.\\
    Sun--Wang; Li et al.\ \cite{sun-wang,li-sun-wang-woodruff} &
    Finite-field addition & Randomized and quantum rank lower bounds;
    sharp adjacent-rank distinctions.\\
    Sherstov--Storozhenko~\cite{sherstov-storozhenko} & Finite-field rank
    approximation & Strong quantum lower bounds under rank promises.\\
    Razborov~\cite{razborov} & Disjointness & Via word packing and a
    matrix encoding, randomized $\Omega(n^2k)$ for rational XOR.\\
    Zhang--Shi~\cite{zhang-shi-xor} & Symmetric XOR functions & Via an
    exact-distance encoding, entangled quantum $\Omega(n^2)$ for Boolean
    rational XOR.\\
    Vempala--Wang--Woodruff~\cite{vempala-wang-woodruff} & Distributed
    linear systems and optimization & Modular feasibility tests;
    feasibility and solution lower bounds for distributed constraints.\\
    \bottomrule
  \end{tabular}
  \caption{Results related to rational singularity.}
  \label{tab:prior-work}
\end{table}

\subsection{Our results}

Let $R^{\rm pub}_{1/3}$
denote public-coin randomized communication in bits and $Q^*_{1/3}$
quantum communication in qubits with arbitrary input-independent
prior entanglement, with
worst-case error at most $1/3$.

We first state our lower and upper bounds for additive singularity.

\begin{samepage}
  \begin{theorem}[Additive lower bound]\label{thm:additive-lower-bound}
    For every $k\ge1$ and sufficiently large $n$,
    \[
      Q^*_{1/3}(\SINGplus_{n,k})\ge
      \Omega\!\left(\min\{n^{3/2}\sqrt{k\log n},n^2\log(nk)\}\right).
    \]
    The bound holds even when the sum is promised to have rank $n-1$ or $n$.
  \end{theorem}
\end{samepage}

\begin{theorem}[Additive upper bound]\label{thm:additive-upper-bound}
  For every $k\ge1$ and sufficiently large $n$,
  \[
    R^{\rm pub}_{1/3}(\SINGplus_{n,k})
    \le O\!\left(n^2\min\{k,\log(k+\log n)\}\right).
  \]
  The protocol always accepts singular matrices (one-sided error).
\end{theorem}

The lower and upper bounds are proved in Sections~\ref{sec:additive-lower-bound}
and~\ref{sec:additive-upper-bound}, respectively. Since $Q^*_{1/3}\le R^{\rm pub}_{1/3}$,
these theorems give lower and upper bounds in both models.

\begin{table}[h]
  \centering
  \small
  \setlength{\tabcolsep}{4pt}
  \begin{tabular}{@{}ccccc@{}}
    \toprule
    $k$ & Previous LB & Our LB & Our UB & Previous UB\\
    \midrule
    $1$ & -- & $\Omega(n^{3/2}\sqrt{\log n})$ & -- & $O(n^2)$ [trivial]\\[4pt]
    $\log n$ & -- & $\Omega(n^{3/2}\log n)$ & $O(n^2\log\log n)$ & $O(n^2\log n)$ [trivial]\\[4pt]
    $\ge n\log n$ & -- & $\Omega(n^2\log(nk))$ & -- &    $O(n^2\log(nk))$~\cite{chu-schnitger}\\
    \bottomrule
  \end{tabular}
  \caption{Bounds for both $R^{\rm pub}_{1/3}(\SINGplus_{n,k})$
    and $Q^*_{1/3}(\SINGplus_{n,k})$ for representative values of $k$.
    A dash in Previous LB means that our literature review did not identify
    a directly matching published lower bound; a dash in Our UB means
  that our bound does not improve the previous upper bound.}
  \label{tab:additive-bounds}
\end{table}

Table~\ref{tab:additive-bounds} displays three representative regimes.
For Boolean inputs the quantum lower bound is already superlinear in
the matrix order, though it does not reach the quadratic input length.
At logarithmic bit length the lifting protocol improves the standard
random-prime upper bound by a factor of order
$\log n/\log\log n$. At $k\ge n\log n$, the lower and upper bounds
match at $\Theta(n^2\log(nk))$. This last assertion includes
superpolynomial bit lengths: it is not restricted to
$k\le n^{O(1)}$. In this entire range prior entanglement and quantum
communication do not change the asymptotic cost.

Our quantum lower bound for XOR holds uniformly over all bit lengths:
\begin{theorem}[Bitwise-XOR quantum lower bound]\label{thm:xor-lower-bound}
  For every $k\ge1$ and sufficiently large $n$,
  \[
    Q^*_{1/3}(\SINGxor_{n,k})\ge\Omega(n^2k).
  \]
  The lower bound holds under the promise that $A\oplus B$ has
  rank $n-1$ or $n$.
\end{theorem}
The proof is in Section~\ref{sec:xor-singularity}. Together with the trivial upper bound,
this gives
\[
  Q^*_{1/3}(\SINGxor_{n,k})=R^{\rm pub}_{1/3}(\SINGxor_{n,k})=\Theta(n^2k).
\]
The lower bound on $R^{\rm pub}_{1/3}(\SINGxor_{n,k})$ already follows from disjointness, and the
case $k=1$ follows from Zhang--Shi's theorem as described above.
The contribution of Theorem~\ref{thm:xor-lower-bound} is the quantum lower
bound with its full linear dependence on growing $k$; the classical
and Boolean statements are included as consequences of earlier tools.
Unlike the additive lower bound, the XOR lower
bound for general $k$ uses overlapping shares; it does not impose
fixed ownership. Indeed, on disjoint supports addition and XOR are
identical, so such a restriction could not yield $\Omega(n^2k)$ for
large $k$ in view of the additive randomized upper bound.

\subsection{The role of the entry bit length}

The separation between addition and XOR is most visible when the
coefficient length grows faster than the matrix order. At $k=n^2$,
for example, both bounded-error complexities of addition are
$\Theta(n^2\log n)$, whereas both complexities of XOR are
$\Theta(n^4)$. At $k=2^n$, the corresponding orders are
$\Theta(n^3)$ and $\Theta(n^2 2^n)$. These are consequences of the
uniform parameter ranges of the theorems, rather than separate
assumptions about how the entries are represented. In each case an
entry is an ordinary binary word, and the players may perform
arbitrarily expensive local computations on it. The savings for
addition concern information transmitted, not arithmetic running time.

There is a simple algebraic reason that the additive upper-bound idea
does not immediately apply to XOR. Reduction modulo an odd prime
commutes with addition: from $A\bmod p$ and his own input, Bob can
form $(A+B)\bmod p$. Bitwise XOR is linear over the individual binary
coordinates, but this does not make its integer value linear modulo
an odd prime. For nonnegative words,
$ a\oplus b=a+b-2(a\mathbin{\&}b) $,
where $\mathbin{\&}$ denotes bitwise AND. The last term depends on
both private words. Sending only the residue of Alice's word does
not in general determine it. This observation is an obstruction to
that particular protocol, not a lower-bound argument; the dense
equation witness is what proves that every bounded-error protocol
must communicate $\Omega(n^2k)$.

At the Boolean endpoint, it is also useful to separate the input
alphabet from the entries of the combined matrix. A general sum of
two Boolean matrices may contain twos, whereas their XOR is Boolean.
The additive lower bound already holds on disjoint supports, where
the sum itself is Boolean. Thus the hardness of additive Boolean
singularity does not arise from the occurrence of twos in the sum.
Conversely, the stronger Boolean XOR lower bound uses private
overlap and does not by itself improve the additive theorem.
The local reductions in the paper specify this ownership information
along with the determinant identity.

Finally, a small output should not be confused with a small
certificate available to one party. A matrix with zero determinant has a
nonzero kernel vector, but neither player generally knows that vector from
their local matrix, and a rational kernel vector can have large coefficients.
Our protocols and reductions therefore account for how the needed
information is obtained. The prime-power protocol sends an explicit
sequence of modular relations, and the lower-bound witnesses control
all short interactive protocols without assuming a particular form
of certificate. The adjacent-rank promises ensure that these issues
remain even when the two possible ranks differ by only one.

\subsection{Proof ideas}

\paragraph{The additive upper bound: divide an implicit determinant.}
A direct modular test fails whenever the chosen prime divides a
nonzero determinant. We instead maintain an implicit integral matrix
obtained by unimodular column operations followed by divisions of
individual columns by powers of that prime. When the matrix is
singular modulo $p$, Bob sends a kernel relation. A carefully chosen
pivot makes the associated column operation integral and unimodular.
The new column is divisible by one more factor of $p$.

Only its residue after division is needed for the next step. For
additive integer shares $a+b$ known to be divisible by $q$, the carry
in the quotient can be determined from Bob's remainder alone. Alice
therefore sends a single quotient residue per coordinate, rather
than the high-precision column. Each lift uses $O(n\log p)$ bits
and removes exactly one factor of $p$ from the determinant. Allowing
$3n$ lifts lets the protocol sample among only $O(k+\log n)$ primes.
The sum of their valuations in a nonzero determinant is bounded by
the determinant's bit length, so fewer than one third exceed the budget.

\paragraph{The additive lower bound: separate height from dimension.}
The lower bound starts with a block matrix containing an $m\times m$
identity and a $d\times d$ Schur complement. A signed witness for
this integer problem has constant correlation with singularity and
operator norm exponentially small in $md\log(u+1)$, where a scalar
cutoff $u$ is adjustable independently of $d$. The entangled quantum
factorization theorem of Linial and Shraibman~\cite[Lemma~12]{linial-shraibman}
turns this norm estimate into a communication lower bound.
The witness has only a small amount of variation on coranks greater
than one, which is charged to the protocol's error. Deleting those
weights would not preserve the norm estimate.

To fit the integer inputs into a smaller alphabet, expand only the
bottom $d$ rows and rightmost $d$ columns into digits. A public
recurrence reconstructs signed differences by elimination, using
entries below $2^k$. It preserves corank and assigns every private
digit to a fixed owner. Thus the dimension cost is proportional to
$d$ times the number of digits, not to the number of private entries.
For short entries take $u$ proportional to $d$ and balance the core
size against the digit expansion. For long entries take $d$ proportional
to $n$ and $u$ proportional to $k/\log k$. This makes all witness
entries fit directly and yields the entire tight long-entry regime.

The witness itself requires control of every Fourier frequency.
Differences of rational isometries give a signed integer residual
correlated with singularity. A positive law on short products of
integral shears has exact prescribed actions on small finite modules;
it supplies Fourier zeros at modular matrices that factor through
few columns. A scalar interpolation filter makes the derivatives
vanish there as well. Dilation by $\operatorname{lcm}(1,\ldots,u)$
extends the zero set to matrices with a small image group. Finally,
a frequency dichotomy says that either the residual transform is
small or the independent free rows strongly suppress the communication
operator. We state the interfaces before proving the algebraic
construction, so the parameter calculation can be read independently
of its supporting lemmas.

\paragraph{The XOR lower bound: one dense equation.}
The classical XOR lower bound follows from the published disjointness
theorem~\cite{razborov} by packing inputs into
words and computing the sum of their bitwise intersections with a
small block determinant. This route alone does not give the desired
quantum bound. For Boolean inputs, an exact Hamming-distance test in
the interior of the range does: Zhang--Shi's symmetric-XOR theorem
\cite[Theorem~2]{zhang-shi-xor} supplies a quadratic quantum lower
bound after a local matrix reduction.

To obtain the factor $k$ for arbitrary word length, we use a dense
weighted equation $a^TXb=\sum Y_{ij}$. The weights factor into a row
coefficient and a column coefficient, so a matrix of linear dimension
can evaluate the equation by a Schur complement. A signed source
compares the exact equation with a shifted equation. Independent
words of $X$ give a product of $d^2$ Fourier factors. The public
coefficients are chosen so this product is exponentially small away
from short neighborhoods of dyadic rational frequencies. At the
centers of these neighborhoods the shift already cancels; a smooth
Fourier cutoff removes the neighborhoods at small total variation cost.

The remaining issue is uniformity in the relation between dimension
and bit length. The scalar output can require extra bits when the
dimension grows. We split it into $O(d^2)$ bounded words whose sum is
the output, and bound every Walsh coefficient on those words, including
masks involving many of them at once. This keeps every matrix entry
within the word budget without requiring $k\ge C\log n$. Walsh
diagonalization then bounds the norm of the XOR communication witness,
and the same entangled factorization theorem gives $\Omega(n^2k)$.

\subsection{Open questions}

The central remaining question is the bounded-error complexity of
addition for small $k$. For Boolean inputs the bounds are
$\Omega(n^{3/2}\sqrt{\log n})$ and $O(n^2)$; at logarithmic bit
length they are $\Omega(n^{3/2}\log n)$ and $O(n^2\log\log n)$.
Does additive singularity require quadratic communication on Boolean
inputs? More generally, is the prime-power upper bound tight throughout
the short-entry range, or is there a protocol that improves it?
These are questions about unrestricted communication, not only about
the particular witness or modular protocol used here.

The proofs point to two concrete issues. The additive lower bound
pays for the height of an integer witness through digit expansion;
a smaller-height witness with comparable correlation and operator norm
would strengthen the result. The upper bound explicitly communicates
kernel relations and quotient residues, so a protocol using different
information about the shared matrix might save communication beyond
this lifting scheme. A further question is whether randomized and
entangled quantum complexities coincide throughout the unresolved
range. Our common lower bound and classical upper bound do not settle
that possibility. In contrast, XOR is settled in both models at every
bit length, even under the adjacent-rank promise.

\paragraph{Organization.}
Section~\ref{sec:preliminaries} presents local reductions and shared tools.
Section~\ref{sec:additive-upper-bound} proves the additive
upper bound. Section~\ref{sec:additive-lower-bound} proves the additive lower bound,
first from an explicit witness interface and then by constructing
that witness. Section~\ref{sec:xor-singularity} proves the XOR theorem.
All necessary unpublished proofs are in
the main text; there are no technical appendices.
\section{Preliminaries}\label{sec:preliminaries}

\subsection{Local reductions and Schur complements}

A local reduction consists of deterministic maps applied separately
to Alice's and Bob's inputs. It requires no communication and
preserves the error of any protocol composed with it. We track
corank, so a reduction also preserves the promise of full rank
versus corank one. In additive constructions, fixed ownership means
each position has a publicly specified owner and the other party
always places zero there. Such inputs have disjoint supports.
For additive and XOR inputs, identity padding is local: Alice adds
an identity block $I_m$ of size $m\times m$ and Bob adds zero,
with zero off-diagonal blocks. We omit the subscript of $I_m$ when
the order is clear.

We repeatedly use Schur complements. If $D$ is invertible, then
\begin{equation}\label{eq:general-schur}
  \rank
  \begin{pmatrix}D&B\\C&E
  \end{pmatrix}
  =\operatorname{ord}(D)+\rank(E-CD^{-1}B),\quad
  \det
  \begin{pmatrix}D&B\\C&E
  \end{pmatrix}
  =\det D\det(E-CD^{-1}B).
\end{equation}
These follow by invertible block row elimination. Two useful
consequences, for blocks of order $d$, are
\begin{align}
  \rank
  \begin{pmatrix}P&Q\\I&I
  \end{pmatrix}
  &=d+\rank(P-Q),&
  \det
  \begin{pmatrix}P&Q\\I&I
  \end{pmatrix}&=\det(P-Q),
  \label{eq:difference-block}\\
  \rank
  \begin{pmatrix}U&V&I\\W&Z&I\\I&I&0
  \end{pmatrix}
  &=2d+\rank(U-V-W+Z).&&\label{eq:four-block}
\end{align}
The determinant in~\eqref{eq:four-block} is
$(-1)^d\det(U-V-W+Z)$. To see the latter identities, substitute
$y'=x+y$ in the block columns and subtract the second block row
from the first. The resulting matrix is
\[
  \begin{pmatrix}U-V-W+Z&V-Z&0\\W-Z&Z&I\\0&I&0
  \end{pmatrix}.
\]
Its lower-right block is invertible with determinant $(-1)^d$;
its Schur complement is $U-V-W+Z$. These identities hold for
arbitrary rational blocks.

Sun--Wang~\cite[p.~478]{sun-wang} give a block reduction from additive
singularity to a balanced row partition, increasing the order from $n$
to $2n$. Their construction preserves corank but uses a $-I_n$ block.
This preserves the alphabet over their finite field, whereas over
$\mathbb Q$ it introduces $-1$, outside our nonnegative $k$-bit alphabet.
The following lemma adapts the construction using only original input
entries, zeros, and ones. It preserves both corank and bit length,
including for Boolean inputs, at the cost of increasing the order to $3n$.

\begin{lemma}[Addition and balanced partitions]\label{lem:addition-column-partition}
  Let $n$ be even and $k\ge1$. The problem $\SINGpart_{n,k}$ reduces
  locally to $\SINGplus_{n,k}$, which in turn reduces locally to
  $\SINGpart_{3n,k}$. Both reductions preserve the entry alphabet and corank.
  Consequently, for deterministic, randomized, or quantum communication
  with the same error and resource conventions,
  \[
    C(\SINGpart_{n,k})\le C(\SINGplus_{n,k})
    \le C(\SINGpart_{3n,k}),
  \]
  where $C$ denotes communication complexity in the chosen model.
  These inequalities also hold under the promise
  of full rank versus corank one.
\end{lemma}
\begin{proof}
  For column inputs $X,Y$ of size $n\times(n/2)$, use additive
  inputs $[\,X\ \ 0\,]$ and $[\,0\ \ Y\,]$; their sum is
  $[\,X\ \ Y\,]$.

  Conversely, for square additive inputs $X,Y$ of order $n$, put
  \[
    M(X,Y)=
    \begin{pmatrix}
      X&Y&0\\ I_n&0&I_n\\ 0&I_n&I_n
    \end{pmatrix}.
  \]
  Its kernel consists exactly of the vectors $(-w,-w,w)$ with
  $(X+Y)w=0$. Thus
  \[
    \rank M(X,Y)=2n+\rank(X+Y),
    \qquad \corank M(X,Y)=\corank(X+Y).
  \]
  Alice constructs the first $n$ columns and Bob the next $n$;
  each block depends only on its owner's input. The last $n$ columns
  are public, so assign half to each party. A fixed column permutation
  places Alice's $3n/2$ columns first and Bob's $3n/2$ columns second.
  All entries are original input entries, zero, or one, so the same
  $k$-bit alphabet suffices. The maps require no communication and
  preserve corank, proving all assertions.
\end{proof}

For odd additive order $n$, append a public $1\times1$ identity block
to $M(X,Y)$ and split its $n+1$ public columns equally. This gives a
balanced column-partition instance of order $3n+1$ with the same
corank and entry alphabet.

\begin{lemma}[Deterministic bounds; consequence of~\cite{chu-schnitger}]
  \label{lem:deterministic-singularity}
  For every $k\ge1$ and sufficiently large $n$,
  \[
    D(\SINGplus_{n,k})=D(\SINGxor_{n,k})=\Theta(n^2k).
  \]
  The lower bounds hold even when the resulting matrix is promised
  to have rank $n-1$ or $n$.
\end{lemma}
\begin{proof}
  For $k\ge2$, the column-partition lower bound of
  Chu--Schnitger~\cite[Theorem~1.1]{chu-schnitger} transfers by placing
  each player's owned columns in its local matrix and zeros elsewhere:
  both the sum and XOR equal the original matrix. For $k=1$, apply
  the two-digit row expansion in the proof of Lemma~\ref{lem:radix-encoding}
  to every row of their two-bit hard instances, with the subtracted
  rows set to zero. The resulting Boolean matrix has five times the
  order and preserves corank and ownership. Assign public entries to
  Alice; the sum and XOR again coincide with the encoded matrix.
  Their hard matrices have corank at most one by the proof of
  Lemma~3.2 and Lemma~3.4 in~\cite{chu-schnitger}, giving the promise.
  Identity padding handles intermediate orders, and sending one input
  gives the upper bound.
\end{proof}

\subsection{Signed witnesses and entangled protocols}
A real signed matrix $W$ has variation $\|W\|_1=\sum_{x,y}|W_{xy}|$;
there is no factor $1/2$ in this convention. For positive probability
vectors $\alpha,\beta$ on the local input sets define
$T_W=D_\alpha^{-1/2}WD_\beta^{-1/2}$. All unadorned matrix norms
in witness estimates are Euclidean operator norms.

\begin{lemma}[Quantum pairing]\label{lem:quantum-witness-pairing}
  If $p_{xy}$ is the acceptance probability of a $t$-qubit protocol with
  arbitrary prior entanglement, then
  \begin{equation}\label{eq:quantum-witness-pairing}
    |\langle p,W\rangle|\le 2^t\|T_W\|.
  \end{equation}
  If $W$ has mass zero and $f$ is the sign of the desired answer, worst-case
  error at most $1/3$ implies
  $\langle p,W\rangle\ge\langle f,W\rangle/2-\|W\|_1/3$.
  If this error guarantee is required only on a promise set and the variation
  outside that set is $\delta$, the same bound holds with an additional
  penalty $2\delta/3$.
\end{lemma}
\begin{proof}
  Linial and Shraibman~\cite[Lemma~12 and Section~3.2.1]{linial-shraibman}
  factor $p_{xy}=\langle a_x,b_y\rangle$ with
  $\|a_x\|,\|b_y\|\le2^{t/2}$, also with prior entanglement. Apply the
  operator-norm inequality in each vector coordinate and Cauchy--Schwarz:
  \[
    |\langle p,W\rangle|\le\|T_W\|
    \left(\sum_x\alpha_x\|a_x\|^2\right)^{1/2}
    \left(\sum_y\beta_y\|b_y\|^2\right)^{1/2}\le2^t\|T_W\|.
  \]
  Pair $p-(1+f)/2$ with $W$. Its magnitude is at most $1/3$ on the
  promise and at most one elsewhere, giving the stated penalty.
  For a promise protocol, cap its execution at its communication budget
  and prescribe an arbitrary output outside the promise. This defines an
  acceptance matrix on the full product without changing promised inputs.
\end{proof}
\section{Additive upper bound}\label{sec:additive-upper-bound}

This section proves Theorem~\ref{thm:additive-upper-bound}.

The standard random-prime protocol~\cite[p.~396]{chu-schnitger}
sends $A$ modulo a prime chosen
among $O(n(k+\log n))$ primes. A nonzero determinant has at most
$n(k+1+\tfrac12\log n)$ distinct prime factors by Hadamard's inequality,
so a sufficiently large constant in the sampling range bounds the
error by $1/3$. Primes of polynomial size in their index give cost
$O(n^2\log(nk))$; see the bounds of Rosser and
Schoenfeld~\cite[Theorem~3 and its corollary]{rosser-schoenfeld}. Combining this with input
transmission gives the previous $O(n^2\min\{k,\log(nk)\})$ bound.
We improve the sampling range by allowing the chosen prime to divide
the determinant up to $3n$ times.
Successive division by $p$ is also central to Dixon's method for exact
linear solving~\cite[pp.~138--139]{dixon}. Dixon lifts a solution using a
matrix invertible modulo $p$. Here we instead remove factors of $p$
from the determinant when the matrix is singular modulo $p$, using
quotient residues that can be communicated from additive shares.

\subsection{Quotient residues from additive shares}
Suppose $a,b\in\mathbb Z$ are private shares and public $q\ge1$
divides $a+b$. Write the Euclidean divisions
$a=qa_0+r_a$, $b=qb_0+r_b$, with $0\le r_a,r_b<q$. The sum of the
remainders is zero or $q$, and it is zero exactly when $r_b=0$.
Therefore
\begin{equation}\label{eq:quotient-from-shares}
  (a+b)/q=a_0+b_0+\mathbf1[r_b\ne0].
\end{equation}
For a public prime $p$, Alice sends only $a_0\bmod p$ and Bob
recovers the quotient modulo $p$. This works for negative shares.
In particular, for a public integer column $v$ with $Mv\in q\mathbb Z^n$,
where $M=A+B$, Bob learns $Mv/q\bmod p$ with
$n\lceil\log p\rceil$ bits by applying~\eqref{eq:quotient-from-shares} to $Av$ and $Bv$.

\subsection{Lifting modular kernel relations}
\begin{lemma}\label{lem:prime-power-lifting}
  For a prime $p$ and a nonnegative integer budget $L$, a deterministic
  protocol of cost $O((n^2+nL)\log p)$ returns nonsingular exactly when
  $\det M\ne0$ and $v_p(\det M)\le L$.
\end{lemma}
\begin{proof}
  Maintain public columns $v_i\in\mathbb Z^n$ and exponents $s_i\ge0$.
  Writing $V=[v_1\ \cdots\ v_n]$, the invariant is
  \begin{equation}\label{eq:lift-invariant}
    \det V=1,\qquad Mv_i\in p^{s_i}\mathbb Z^n,\qquad
    R=MV\diag(p^{-s_1},\ldots,p^{-s_n})\in\mathbb Z^{n\times n}.
  \end{equation}
  Bob knows $R\bmod p$. Initially $V=I$, $s_i=0$, and Alice sends
  $A\bmod p$, costing $n^2\lceil\log p\rceil$ bits.

  If $R\bmod p$ is invertible, return nonsingular. Otherwise Bob chooses
  a nonzero $c\in\ker(R\bmod p)$. Among its nonzero coordinates choose
  $j$ with the largest $s_j$, breaking ties by the smallest index. Normalize
  $c_j=1$ and represent every coordinate in $\{0,\ldots,p-1\}$.
  If the lift budget is exhausted, return singular. Otherwise Bob sends
  $c$; both parties recover $j$ from its support and the public exponents.
  With $h=s_j$, replace
  \begin{equation}\label{eq:lift-update}
    v_j\ \longleftarrow\ \sum_{i:c_i\ne0}c_i p^{h-s_i}v_i,
    \qquad s_j\ \longleftarrow\ h+1.
  \end{equation}
  The exponents in the sum are nonnegative and the old $v_j$ has
  coefficient one. Thus the column operation preserves $\det V=1$.
  Using the old $R$, the new column satisfies
  $Mv_j/p^h=\sum_i c_iR_{:,i}\in p\mathbb Z^n$.
  The quotient protocol~\eqref{eq:quotient-from-shares}, with $q=p^{h+1}$, sends
  its new residue after division. Other columns of $R\bmod p$ are
  unchanged. Each lift costs two vectors of $n$ residues plus a control
  bit. Public integers in $V$ are computed locally from the transcript;
  their bit lengths are not communicated.

  After $t$ lifts, $\sum_i s_i=t$ and
  \begin{equation}\label{eq:lift-determinant}
    \det R=\det M/p^t.
  \end{equation}
  For nonzero $\det M$, the integral matrix $R$ is singular modulo $p$
  exactly while $t<v_p(\det M)$. Thus the protocol stops after precisely
  that many lifts if the budget allows, including the final check at
  $t=L$. If $\det M=0$, the implicit matrix stays singular over
  $\mathbb Q$ and modulo $p$. These observations prove correctness
  and the worst-case communication bound.
\end{proof}

For example, $M=\bigl(
  \begin{smallmatrix}2&1\\0&2
\end{smallmatrix}\bigr)$
and $p=2$ require two lifts. First divide the first column by $2$.
The two columns then agree modulo $2$, and~\eqref{eq:lift-update}
replaces $v_1$ by $e_1+2e_2$ and raises its exponent to two.
The implicit matrix becomes
$\bigl(
  \begin{smallmatrix}1&1\\1&2
\end{smallmatrix}\bigr)$,
which is invertible modulo $2$. Each lift removes one of the two
prime factors in $\det M=4$.

\subsection{Prime sampling and communication cost}
\begin{proof}[Proof of Theorem~\ref{thm:additive-upper-bound}]
  Put $T=k+1+\lceil\log(n+1)\rceil$ and choose a public uniform prime
  among the first $T$ primes. Run Lemma~\ref{lem:prime-power-lifting} with $L=3n$.
  The standard bound $p_T=O(T\log(T+1))$~\cite[Theorem~3 and its corollary]{rosser-schoenfeld}
  gives $O(\log T)$ bits per residue and hence $O(n^2\log T)$ total bits.
  For $\det M\ne0$, Hadamard's inequality and unique factorization give
  \[
    \sum_{i=1}^T v_{p_i}(\det M)
    \le\sum_{i=1}^T v_{p_i}(\det M)\log p_i
    \le\log|\det M|
    \le n(k+1+\tfrac12\log n)\le nT.
  \]
  The protocol errs only when the selected valuation exceeds $3n$;
  the fraction of such primes is at most $n/(3n+1)<1/3$.
  Singular matrices are always accepted. Sending $A$ gives the competing
  $n^2k$ upper bound, so choosing the cheaper protocol proves the result.
\end{proof}
\section{Additive lower bound}\label{sec:additive-lower-bound}

Our construction establishes the lower bound under an additional
ownership restriction. For each $n,k$, it fixes a public owner for
every matrix position, independently of the private inputs, and the
other player's entry is zero.

To prove Theorem~\ref{thm:additive-lower-bound}, we first isolate
the hard integer input and the digit encoding,
derive the stated parameters, and then prove the witness interface.
This separates the communication reduction from the algebra needed
to construct its signed weights.

\subsection{Hard integer inputs and their quantum witness}\label{sec:hard-integer-inputs}
Alice holds $P=[P_0\mid P_1]$ and Bob holds $Z=[Z_0\mid Z_1]$,
both of size $d\times(m+d)$. They hold separate, publicly assigned
rows of $F\in\mathbb Z^{m\times d}$. Define
\begin{equation}\label{eq:additive-core}
  C=
  \begin{pmatrix}I_m&F\\P_0-Z_0&P_1-Z_1
  \end{pmatrix},
  \qquad E=P_1-Z_1+(Z_0-P_0)F.
\end{equation}
Block elimination gives
\begin{equation}\label{eq:additive-core-rank}
  \rank C=m+\rank E.
\end{equation}
Every large entry lies in the bottom $d$ rows or the rightmost $d$
columns, which will make the digit encoding economical.

\begin{proposition}[Hard integer inputs]\label{prop:hard-integer-inputs}
  For every sufficiently large power of two $d$ and all integers
  $m,u\ge1$, there are finite local input sets in~\eqref{eq:additive-core}
  with nonnegative entries of at most
  \begin{equation}\label{eq:integer-height}
    h=O\!\left(\log d+(d+u)\log(u+1)+\log m\right)
  \end{equation}
  bits, on which singularity of $C$ requires $\Omega(md\log(u+1))$
  qubits at error $1/3$. The same lower bound holds when
  $\rank C\in\{m+d-1,m+d\}$ is promised.
\end{proposition}

\begin{lemma}[Integer witness]\label{lem:additive-integer-witness}
  For the input sets in Proposition~\ref{prop:hard-integer-inputs}, put
  $s=d/512$ and $B=(u+1)^{s+1}$. There are positive reference
  probabilities $\alpha,\beta$ and real weights $W$ such that
  \begin{equation}\label{eq:additive-witness-properties}
    \sum W=0,\quad \|W\|_1\le1,\quad
    \langle f,W\rangle\ge5/6,\quad
    \|T_W\|\le e^{1/2}(2/B)^{m/2},
  \end{equation}
  where $f=+1$ on singular $C$ and $-1$ otherwise. Moreover,
  $\sum_{\rank C\le m+d-2}|W_{xy}|\le1/40$.
  Every input in the full product of the reference supports satisfies
  the entry bound~\eqref{eq:integer-height}.
\end{lemma}
The construction and all its ingredients are proved in the remaining
subsections, beginning with Section~\ref{sec:additive-witness-construction}.

\begin{proof}[Proof of Proposition~\ref{prop:hard-integer-inputs}]
  Let $p$ be the acceptance matrix of a $t$-qubit protocol. Applying
  Lemma~\ref{lem:quantum-witness-pairing}, including its promise penalty, gives
  \[
    \frac1{15}=\frac5{12}-\frac13-\frac{2}{3\cdot40}
    \le\langle p,W\rangle\le2^t e^{1/2}(2/B)^{m/2}.
  \]
  Hence $t\ge\tfrac m2(\log B-1)-5=\Omega(md\log(u+1))$.
  We keep the full witness; restricting it to the rank promise would
  not justify retaining its operator bound.
\end{proof}
\subsection{Local radix encoding}\label{sec:radix-encoding}

The encoding expands only the bottom $d$ rows and rightmost $d$
columns of~\eqref{eq:additive-core}. Its cost is proportional to $d$ times
the number of digits, rather than to the number of integer entries.

\begin{lemma}[Radix encoding]\label{lem:radix-encoding}
  Suppose $P,Z,F$ in~\eqref{eq:additive-core} have $h$-bit entries,
  where $h\ge1$. For every $k\ge1$ the parties can locally construct
  $k$-bit matrices $A',B'$ with fixed ownership and disjoint supports,
  of order
  \begin{equation}\label{eq:encoded-order}
    n'=m+d(6\lceil h/k\rceil-3)\le m+3d+6dh/k,
    \qquad \rank(A'+B')=n'-d+\rank E.
  \end{equation}
\end{lemma}

\begin{proof}
  First consider two binary digits. Let $a=a_0+2a_1$ and
  $b=b_0+2b_1$, with all four digits Boolean. Put
  \[
    M=\left(
      \begin{array}{c|cccc}
        a_0&1&0&0&1\\b_0&1&0&1&0\\0&0&0&1&1\\
        a_1&0&1&0&0\\b_1&0&1&0&1
    \end{array}\right).
  \]
  The row weights $(1,-1,1,2,-2)$ cancel the four public columns
  and leave $a-b$. These columns are independent, so elimination
  leaves precisely the original difference and four additional pivots.
  The same columns work when each digit is an entire row vector.

  For the general construction put $R=2^k$. To recover a difference
  from $J$ base-$R$ digits, use rows $p_j,z_j$ for $0\le j<J$ and
  auxiliary rows $t_j$ for $0\le j<J-1$. Form a public matrix
  $\Gamma_J$ by making one column from the coefficients of each
  of the following equations in row weights $v$:
  \begin{equation}\label{eq:radix-relations}
    \begin{aligned}
      v_{p_j}+v_{z_j}&=0 &&(0\le j<J),\\
      v_{t_j}+v_{z_j}&=0 &&(0\le j<J-1),\\
      (R-1)v_{p_j}+v_{t_j}+v_{z_{j+1}}&=0 &&(0\le j<J-1).
    \end{aligned}
  \end{equation}
  Thus $v_{z_j}=-v_{p_j}$, $v_{t_j}=v_{p_j}$, and
  $v_{p_{j+1}}=Rv_{p_j}$. The left kernel of $\Gamma_J$ is
  one-dimensional, with normalized vector $w_J$ having weights $R^j$
  on $p_j,t_j$ and $-R^j$ on $z_j$. There are $3J-1$ rows and
  $3J-2$ independent columns. Every entry of $\Gamma_J$ is
  $0$, $1$, or $R-1$: the multiplier $R$ is implemented as $(R-1)+1$.

  Replace a data row $\sum_jR^j(a_j-b_j)$ by the digit rows $a_j,b_j$
  and zero auxiliary rows, and attach $\Gamma_J$ with zeros in the
  other rows. Row elimination turns its public columns into an identity
  block and a zero row; choose $w_J^T$ as the remaining row operation.
  It recovers the original data. Clearing the other data against the
  identity block shows
  \begin{equation}\label{eq:radix-rank}
    \rank(\text{expanded matrix})
    =\rank(\text{original matrix})+3J-2.
  \end{equation}
  The expansion adds equally many rows and columns, so it preserves
  corank. Transposing gives the same rule for columns.

  Use $J$ digits for $P,Z$ and $J'$ digits for $F$. Apply the row
  rule to each bottom row of $C$, placing Alice's digits of $P$ in
  $p_j$ rows and Bob's digits of $Z$ in $z_j$ rows. Only the top
  block $F$ still has large entries. Next apply the transposed rule
  to each of the last $d$ original data columns. Put the top entries'
  digits in their $p_j$ columns. Put every bottom entry only in its
  $p_0$ column, whose weight is one; these entries already fit in
  $k$ bits. All other data positions are zero. Multiplication by
  $w_{J'}$ recovers the entire original column.

  The resulting order is
  \begin{equation}\label{eq:general-radix-order}
    (m+d)+d(3J-2)+d(3J'-2)=m+d(3J+3J'-3).
  \end{equation}
  Take $J=J'=\lceil h/k\rceil$ to obtain~\eqref{eq:encoded-order}.
  Each private digit has a position reserved for its owner, including
  the publicly assigned rows of $F$. Assign every public entry to
  Alice and zero to the other party. The maps are local, their
  supports are disjoint, and all entries are less than $R$.
  Corank preservation and~\eqref{eq:additive-core-rank} prove the rank identity.
\end{proof}

\subsection{Choosing the core size and scalar cutoff}
\begin{proof}[Proof of Theorem~\ref{thm:additive-lower-bound}]
  Take $m=\lfloor n/2\rfloor$. For $1\le k\le n\log n$, set
  $u=d/512$. Proposition~\ref{prop:hard-integer-inputs} gives $h=O(d\log n)$
  for $d\le n$, so Lemma~\ref{lem:radix-encoding} uses at most
  \[
    n'\le n/2+C_1\left(d+\frac{d^2\log n}{k}\right)
  \]
  rows. Let $d$ be the largest power of two not exceeding
  $\eta\min\{n,\sqrt{nk/\log n}\}$, where the absolute constant
  $\eta>0$ is chosen small enough for both ranges of this proof;
  in particular, $\eta\le1/6$ and $C_1(\eta+\eta^2)\le1/2$.
  Then $n'\le n$, and identity padding is local and preserves corank.
  For all sufficiently large $n$, $d$ exceeds the fixed construction
  threshold and $\log d=\Theta(\log n)$ uniformly in $k\ge1$.
  The communication bound is
  $\Omega(md\log d)=\Omega(n^{3/2}\sqrt{k\log n})$.

  For $k\ge n\log n$, choose $d$ as the largest power of two below
  $\eta n$ and put $u=\lfloor\eta k/\log(k+1)\rfloor$.
  For sufficiently large $n$, $d=\Theta(n)$ and
  $\log(u+1)=\Theta(\log k)$. Monotonicity of $\log t/t$ implies
  $n\log(k+1)\le3k$ throughout this range. Therefore
  \[
    d\log(u+1)\le3\eta k,\qquad
    u\log(u+1)\le\eta k.
  \]
  The $O(\log n)$ term in~\eqref{eq:integer-height} is $o(k)$.
  Our choice of $\eta$ also ensures $h\le k$ for sufficiently large $n$.
  The one-digit encoding has order $m+3d\le n$.
  After padding the lower bound is
  $\Omega(md\log(u+1))=\Omega(n^2\log k)=\Omega(n^2\log(nk))$.
  Both ranges imply the stated minimum, with one absolute constant.
  All reductions preserve the rank promise and fixed ownership.
\end{proof}

\begin{corollary}\label{cor:additive-long-entries}
  For every $k\ge n\log n$ and sufficiently large $n$,
  \[
    Q^*_{1/3}(\SINGplus_{n,k})=R^{\rm pub}_{1/3}(\SINGplus_{n,k})
    =\Theta(n^2\log(nk)).
  \]
  This also holds under the adjacent-rank promise.
  For $k=1$, the lower bound is $\Omega(n^{3/2}\sqrt{\log n})$.
\end{corollary}
\begin{proof}
  Combine Theorems~\ref{thm:additive-upper-bound}
  and~\ref{thm:additive-lower-bound}.
  For $k\ge n\log n$, $\log(k+\log n)=\Theta(\log(nk))$.
\end{proof}
\subsection{Constructing the additive communication witness}\label{sec:additive-witness-construction}

We now prove Lemma~\ref{lem:additive-integer-witness}. For a signed measure $\nu$
on integer $d\times d$ matrices, write
\[
  \|\nu\|_1=\sum_E|\nu(E)|,\qquad
  \widehat\nu(\theta)=\sum_E\nu(E)
  e^{2\pi i\langle\theta,E\rangle},\qquad
  \langle\theta,E\rangle=\sum_{i,j}\theta_{ij}E_{ij}.
\]
Let $g_d(E)=+1$ for singular $E$ and $-1$ otherwise.
The entry height is $\max_{E\in\supp\nu,i,j}|E_{ij}|$.
The notation $\|\Delta\|_{\entry}$ denotes the sum of the absolute
values of all entries. Pushing forward a signed measure means
summing its signed weights over the fibers of a map; no positivity
is assumed.

\subsubsection{Structured residuals and scalar dilation}

The remaining subsections establish the following input to the construction.

\begin{lemma}[Structured residual]\label{lem:residual-interface}
  Let $d=2r$, where $r$ is a power of two, and suppose
  $r\ge160s$, $B\ge3$, and $L\ge1$. There is a finite rational signed
  measure $\nu_0$ on integer $d\times d$ matrices with mass zero,
  variation at most one, correlation
  $\langle g_d,\nu_0\rangle\ge403/440$, and height $K_0$ satisfying
  \begin{equation}\label{eq:residual-height}
    \log(K_0+1)=O(\log d+\log B+\log(L+1)).
  \end{equation}
  Its variation on coranks at least two is at most $1/40$. All Fourier derivatives of orders $0,1,\ldots,L$ vanish at $Q/t$
  whenever $t\le B$ and $Q$ factors through at most $s$ columns
  over $\Z/t$.
\end{lemma}

The proof is Corollary~\ref{cor:structured-residual}.
Here a factorization through $s$ columns means $Q=UV$ with
$U\in(\Z/t)^{d\times s}$ and $V\in(\Z/t)^{s\times d}$.
The derivative condition ensures that the transform remains small
near its zeros. The next scalar operation makes those zeros cover
every frequency with a small column image group.

\begin{lemma}[Scalar cutoff]\label{lem:scalar-cutoff}
  Let $a$ be a positive integer, and put
  $w(a)=\min_{p\text{ prime}}p^{v_p(a)+1}$.
  If a finite abelian group $G$ has $|G|<w(a)^{s+1}$, then $aG$
  can be generated by at most $s$ elements. In particular, for
  every integer $u\ge1$ and $\lambda_u=\operatorname{lcm}(1,\ldots,u)$,
  \begin{equation}\label{eq:scalar-generator-bound}
    |G|<(u+1)^{s+1}\quad\Longrightarrow\quad
    \lambda_uG\text{ has at most }s\text{ generators}.
  \end{equation}
  Moreover, $\log\lambda_u\le u\log(u+1)$.
\end{lemma}

\begin{proof}
  Write the $p$-primary part of $G$ as $\bigoplus_i\Z/p^{e_i}$.
  Multiplication by $a$ leaves a nonzero cyclic factor precisely when
  $e_i>v_p(a)$. If at least $s+1$ factors survived, then
  $|G|\ge p^{(v_p(a)+1)(s+1)}\ge w(a)^{s+1}$, a contradiction.
  The number of generators of a finite abelian group is the maximum
  of its primary generator counts: pad each list to the same length
  and combine them by the Chinese remainder theorem. This proves the
  first claim. Since $p^{v_p(\lambda_u)+1}>u$ for every prime $p$,
  $w(\lambda_u)\ge u+1$. Finally $\lambda_u\mid u!$.
\end{proof}

Take $B=(u+1)^{s+1}$ and push $\nu_0$ forward under
$E\mapsto\lambda_uE$, obtaining $\nu$. This preserves mass,
variation, and rational rank. Its height $K=\lambda_uK_0$ obeys
\begin{equation}\label{eq:dilated-height}
  \log(K+1)=O(\log d+\log B+\log(L+1)+u\log(u+1)).
\end{equation}
Let $\operatorname{im}_t Q$ denote the column image group of $Q$
over $\Z/t$. If $t\le B$ and $|\operatorname{im}_t Q|<B$,
Lemma~\ref{lem:scalar-cutoff} shows that $\lambda_uQ$ factors through
$s$ columns: take generators of its image as the columns of $U$.
The chain rule and Lemma~\ref{lem:residual-interface} therefore
make every derivative through order $L$ vanish at $Q/t$.
Taylor's theorem gives
\begin{equation}\label{eq:residual-taylor-bound}
  |\widehat\nu(Q/t+\Delta)|\le2^{-(L+1)}
  \quad\text{when}\quad
  \|\Delta\|_{\entry}\le\frac{L+1}{40K}.
\end{equation}
Indeed, the remainder is at most
$(2\pi K\|\Delta\|_{\entry})^{L+1}/(L+1)!$, and
$(L+1)!\ge((L+1)/e)^{L+1}$ with $2\pi e<20$.

\subsubsection{A dichotomy for torus frequencies}

For $v\in\T^d=(\mathbb R/\Z)^d$, set
$\|v\|_{\T,\infty}=\max_i\operatorname{dist}(v_i,\Z)$.
We first need an elementary estimate for finite subgroups.

\begin{lemma}\label{lem:finite-subgroup-small-ball}
  If $G\le\T^d$ has order $M$ and $0\le\rho<1/4$, then
  \[
    \Pr_{x\sim\Unif(G)}[\|u+x\|_{\T,\infty}\le\rho]
    \le\max\{4\rho,1/M\}
    \qquad(u\in\T^d).
  \]
\end{lemma}

\begin{proof}
  Let $G_i$ be the subgroup of elements whose first $i$ coordinates are zero,
  and put $q_i=|G_{i-1}|/|G_i|$. Conditional on any attainable
  first $i-1$ coordinates, coordinate $i$ is uniform on a coset of
  the cyclic subgroup of order $q_i$. An arc of length $2\rho$
  contains at most $\lfloor2\rho q_i\rfloor+1$ such points.
  If every $q_i<1/(2\rho)$, there is at most one continuation at
  each step, giving probability at most $\prod_iq_i^{-1}=1/M$.
  Otherwise, a coordinate with $q_i\ge1/(2\rho)$ has conditional
  probability at most $2\rho+1/q_i\le4\rho$. Other restrictions
  can only decrease the probability. When $\rho=0$, the event
  contains at most one group element.
\end{proof}

Let the coordinates of a random column $f\in\Z^d$ be independent,
with
\begin{equation}\label{eq:free-row-law}
  f_j=V_j+\sum_{t=1}^B R_{jt},\qquad
  V_j\sim\Unif\{0,\ldots,4B-1\},\qquad
  R_{jt}\sim\Unif\{0,\ldots,t-1\}.
\end{equation}
All variables are independent. Each coordinate is uniform modulo
every $t\le B$ and lies in $[0,T]$, where
$T=4B-1+B(B-1)/2$. For $\theta\in\T^{d\times d}$ put
\[
  C_H(\theta)=\sup_{v\in\T^d}
  \Pr[\|v+\theta f\|_{\T,\infty}\le1/H].
\]

\begin{lemma}[Frequency dichotomy]\label{lem:frequency-dichotomy}
  If
  \begin{equation}\label{eq:mask-width}
    H\ge\max\{2,\,8B(1+2dT),\,80d^2K/(L+1)\},
  \end{equation}
  then every $\theta\in\T^{d\times d}$ satisfies either
  $C_H(\theta)\le1/B$ or $|\widehat\nu(\theta)|\le2^{-(L+1)}$.
\end{lemma}

\begin{proof}
  Suppose some translated ball has probability greater than $1/B$.
  For each $j$, condition on all variables except $V_j$. Some fiber
  has at least five good values among $0,\ldots,4B-1$. Two consecutive
  good values have a gap $1\le q_j<B$. Subtracting their conditions
  shows that column $j$ of $\theta$ is within $2/H$ entrywise, on
  the torus, of a rational column $a_j/q_j$. Thus
  $\theta=\theta_0+\Delta$, where $|\Delta_{ij}|\le2/H$.

  Let $G$ be the subgroup of $\T^d$ generated by the columns of
  $\theta_0$, and put $M=|G|$. The coordinates $f_j$ are independent
  and uniform modulo their respective $q_j$, so $\theta_0f$ is
  uniform on $G$. On the original event it belongs to a translated
  ball of radius $\rho=(1+2dT)/H$. Lemma~\ref{lem:finite-subgroup-small-ball}
  bounds the probability by $\max\{4\rho,1/M\}$.
  Since $4\rho\le1/(2B)$, we must have $M<B$.

  Every element of $G$ is annihilated by $M$, so write
  $\theta_0=Q/M$. The numerator map identifies $G$ with
  $\operatorname{im}_M Q$, whose order is $M<B$.
  Consequently~\eqref{eq:residual-taylor-bound} applies, since
  $\|\Delta\|_{\entry}\le2d^2/H\le(L+1)/(40K)$.
  This proves the claim, including $M=1$.
\end{proof}

\subsubsection{Fourier decomposition of the communication operator}

Take $m$ independent rows $f_t^T$ with law~\eqref{eq:free-row-law}
to form $F$. Fix any public partition of these rows between the
parties, and denote the row law by $\pi$. Put
\begin{equation}\label{eq:mask-coefficients}
  \ell=\lceil\log B\rceil,\quad L=\ell m,\quad
  a_{it}=\sum_{j=1}^{\ell}U_{itj},\quad
  U_{itj}\sim\Unif\{0,\ldots,H-1\},\quad
  a_{\max}=\ell(H-1),
\end{equation}
with independent coefficients. Let
\begin{equation}\label{eq:mask-kernel}
  \begin{split}
    D_H(v)&=H^{-1}\sum_{a=0}^{H-1}e^{2\pi iav},\qquad
    \phi(v)=\prod_{i=1}^dD_H(v_i)^\ell,\\
    \kappa_\theta(u)&=\mathbb E_f|\phi(u+\theta f)|^2.
  \end{split}
\end{equation}
The geometric-sum estimate gives
$|D_H(v)|\le\min\{1,(2\operatorname{dist}(v,\Z)H)^{-1}\}$.
Outside the box in the definition of $C_H$, at least one squared
factor is at most $4^{-\ell}$. Hence
\begin{equation}\label{eq:mask-kernel-bound}
  \sup_u\kappa_\theta(u)\le C_H(\theta)+4^{-\ell}.
\end{equation}

First work modulo any integer $q\ge1$. Choose $P$ uniformly in
$(\Z/q)^{d\times(m+d)}$ and push the product of its law, the laws
of $F,a$, and the signed measure $\nu(E)$ to the pair of inputs
specified by
\begin{equation}\label{eq:masked-inputs}
  Z_0=P_0+a,\qquad Z_1=P_1+aF-E\pmod q.
\end{equation}
Call this signed matrix $W_{\circ,q}$. With $h_0=d(m+d)$, use
the positive reference laws
\[
  \alpha_\circ(F_A,P)=\pi^{\otimes m_A}(F_A)q^{-h_0},\qquad
  \beta_\circ(F_B,Z)=\pi^{\otimes m_B}(F_B)q^{-h_0},
\]
where $m_A+m_B=m$. Let $T_{\circ,q}$ be the normalized matrix
$D_{\alpha_\circ}^{-1/2}W_{\circ,q}D_{\beta_\circ}^{-1/2}$.

Unitary Fourier transformation in the $P,Z$ coordinates splits
this operator into blocks indexed by
$u_1,\ldots,u_m\in(q^{-1}\Z/\Z)^d$ and
$\Theta\in(\Z/q)^{d\times d}$. Translation by $(a,aF-E)$
contributes the phase
\[
  \exp\left(2\pi i\left[
      \sum_t a_t^T(u_t+\Theta f_t/q)-\langle\Theta,E\rangle/q
  \right]\right).
\]
After averaging over the independent coefficients and $E$, the block
has entries
\[
  \widehat\nu(-\Theta/q)
  \prod_{t=1}^m\sqrt{\pi(f_t)}\,
  \phi(u_t+\Theta f_t/q).
\]
This product separates between Alice's and Bob's rows. Each block
therefore has rank at most one, with norm
$|\widehat\nu(\Theta/q)|\prod_t\sqrt{\kappa_{\Theta/q}(u_t)}$.
The sign can be dropped because $\nu$ is real. It follows that
\begin{align}
  \|T_{\circ,q}\|
  &=\max_{\Theta\bmod q}|\widehat\nu(\Theta/q)|
  \bigl(\max_u\kappa_{\Theta/q}(u)\bigr)^{m/2}\notag\\
  &\le\max\{2^{-(L+1)},(B^{-1}+4^{-\ell})^{m/2}\}
  \le(2/B)^{m/2}.\label{eq:cyclic-witness-norm}
\end{align}
The inner maximum in the equality is over the period-$q$ grid.
The inequality uses Lemma~\ref{lem:frequency-dichotomy}, \eqref{eq:mask-kernel-bound},
$|\widehat\nu|\le\|\nu\|_1\le1$, $\kappa\le1$,
$4^{-\ell}\le B^{-2}$, and $L=\ell m$.

\subsubsection{Restricting the witness to bounded integer inputs}

Choose integer parameters
\begin{equation}\label{eq:integer-reservoir}
  \begin{split}
    D&=\left\lceil\max\{h_0(ma_{\max}T+2K),1\}\right\rceil,\\
    q&=D+ma_{\max}T+2K,\qquad I=\{K,\ldots,K+D-1\}.
  \end{split}
\end{equation}
Choose $P$ uniformly in $I^{h_0}$ and use~\eqref{eq:masked-inputs}
over the integers. Every entry of $P,Z$ lies in $[0,q)$:
the shift $aF$ is nonnegative and at most $ma_{\max}T$, and
$|E_{ij}|\le K$. Thus there is no modular wrap. Let $W$ be the
resulting signed matrix. It is the restriction of $W_{\circ,q}$
to $P\in I^{h_0}$, multiplied by $(q/D)^{h_0}$.

Normalize Alice's side by her actual probability law
$\alpha(F_A,P)=\pi^{\otimes m_A}(F_A)D^{-h_0}$ and keep
$\beta=\beta_\circ$ on Bob's side. Row restriction cannot increase
the operator norm, and the change of normalization gives
\begin{equation}\label{eq:integer-witness-norm}
  \|T_W\|\le(q/D)^{h_0/2}\|T_{\circ,q}\|
  \le e^{1/2}(2/B)^{m/2},
\end{equation}
since $q/D=1+(ma_{\max}T+2K)/D\le1+1/h_0$.

Pushing forward a signed measure preserves its mass and cannot
increase its variation. Hence $W$ has mass zero and variation at
most one. For every sampled pair, substitution in~\eqref{eq:additive-core}
recovers the residual $E$ exactly. By~\eqref{eq:additive-core-rank}, its
singularity label is $g_d(E)$, and consequently
\[
  \langle f,W\rangle=\langle g_d,\nu\rangle
  \ge403/440>5/6.
\]
The residual's exceptional variation is also preserved as an upper
bound: forming $W$ averages over probability laws, recovers $E$ exactly,
and can only decrease variation by the triangle inequality. Thus
$\sum_{\rank C\le m+d-2}|W_{xy}|\le1/40$.
The full product of the two reference supports consists of valid
integer inputs. The label $f$ is defined there even on pairs that
receive zero witness weight.

\begin{proof}[Completion of Lemma~\ref{lem:additive-integer-witness}]
  Choose
  \[
    r=d/2,\qquad s=d/512,\qquad B=(u+1)^{s+1},\qquad
    \ell=\lceil\log B\rceil,\qquad L=\ell m.
  \]
  For sufficiently large powers of two $d$, these are integer
  parameters with $r\ge160s$. We have
  $\log B=\Theta(d\log(u+1))$, so~\eqref{eq:dilated-height} gives
  $\log(K+1)=O(\log d+(d+u)\log(u+1)+\log m)$.
  Also $\log(T+1)=O(\log B)$.
  Take the least integer $H$ satisfying~\eqref{eq:mask-width}.
  Equations~\eqref{eq:mask-width}, \eqref{eq:mask-coefficients}, and~\eqref{eq:integer-reservoir}
  then give
  \[
    \log H,\ \log(a_{\max}+1),\ \log q
    =O(\log d+(d+u)\log(u+1)+\log m).
  \]
  Thus every entry in either reference support is nonnegative and
  has the required number of bits. The mass, variation, correlation,
  and norm estimates above prove~\eqref{eq:additive-witness-properties}.
\end{proof}

\subsection{Split isometries and singularity correlation}\label{sec:split-isometries}

\subsubsection{A parametrization by two shears}

Let $r\ge2$ be even and put $d=2r$. Write

\[
  J_0=
  \begin{pmatrix}0&I_r\\I_r&0
  \end{pmatrix},\qquad
  J=
  \begin{pmatrix}I_r&0\\0&-I_r
  \end{pmatrix},\qquad
  P=
  \begin{pmatrix}I_r&I_r\\I_r&-I_r
  \end{pmatrix}.
\]

For skew-symmetric $r$-by-$r$ matrices $X,Y$, define

\begin{equation}
  T(X,Y)=
  \begin{pmatrix}I&X\\0&I
  \end{pmatrix}
  \begin{pmatrix}I&0\\Y&I
  \end{pmatrix}
  =
  \begin{pmatrix}I+XY&X\\Y&I
  \end{pmatrix},\qquad
  R(X,Y)=\tfrac12 PT(X,Y)P.                         \label{eq:two-shear-isometry}
\end{equation}

Direct multiplication gives $T^TJ_0T=J_0$, $\det T=1$,
$R^TJR=J$, and $\det R=1$.
\subsubsection{Rank parity for orthogonal isometries}

\begin{lemma}[Rank parity; Dickson invariant]

  \label{lem:isometry-rank-parity}

  Let $g$ preserve a nondegenerate symmetric bilinear form over
  $\mathbb{Q}$. Then $\det g=(-1)^{\operatorname{rank} (I-g)}$.

\end{lemma}
\begin{proof}
  The Dickson invariant is
  $\Delta(g)=\dim\operatorname{im}(I-g)\bmod2$.
  In characteristic different from two it satisfies
  $\det g=(-1)^{\Delta(g)}$; see
  First~\cite[Section~5.1, equation~(5.1)]{first}.
  Apply this identity over $\mathbb Q$ to the quadratic form
  $x\mapsto b(x,x)/2$ associated with the given bilinear form.
\end{proof}

\subsubsection{Uniform transversality for two shears}

\begin{lemma}

  \label{lem:skew-completion}

  Suppose $\binom AC$ has column rank $r$ and $A^TC$ is skew-symmetric.
  There is a rational skew-symmetric $Y$ with
  $\operatorname{rank} (C+YA)\ge r-1$.

\end{lemma}
\begin{proof}

  Choose invertible $S,V$ with $SAV=\operatorname{diag}(I_s,0)$.
  Set $C'=S^{-T}CV$. The skew condition and full column rank give

  \[
    C'=
    \begin{pmatrix}K&0\\ B&E
    \end{pmatrix},
    \qquad K^T=-K,\qquad E\text{ invertible}.
  \]

  Let $F_s$ be a skew-symmetric matrix of rank
  $s-(s\bmod2)$, for example a direct sum of standard alternating
  two-by-two blocks and at most one zero block. Take

  \[
    Y'=
    \begin{pmatrix}F_s-K&B^T\\-B&0
    \end{pmatrix},
    \qquad Y=S^TY'S.
  \]

  Then $C'+Y'\operatorname{diag}(I_s,0)=\operatorname{diag}(F_s,E)$,
  of rank at least $r-1$. Transforming back proves the claim.

\end{proof}

\begin{theorem}

  \label{thm:two-shear-transversality}

  For every fixed rational $J_0$-isometry $G$, there are rational
  skew-symmetric $X,Y$ such that

  \begin{equation}
    \operatorname{rank} (I-GT(X,Y))=
    \begin{cases}
      d,&\det G=1,\\
      d-1,&\det G=-1.
    \end{cases}                                             \label{eq:transverse-residual-rank}
  \end{equation}

  If the free entries of $X,Y$ are instead independently uniform in
  an arbitrary $L$-element subset of $\mathbb{Q}$, the probability of
  failing to attain the indicated rank is at most $2d/L$.

\end{theorem}
\begin{proof}

  Write $G=
  \begin{pmatrix}A&B\\ C&D
  \end{pmatrix}$. Its first $r$
  columns satisfy Lemma~\ref{lem:skew-completion}, since $G^TJ_0G=J_0$.
  Choose $Y$ from that lemma and a fixed invertible skew-symmetric
  $F$ of order $r$, which exists because $r$ is even.
  Right multiplication by the lower shear with parameter $-Y$ gives

  \begin{equation}
    (I-GT(tF,Y))
    \begin{pmatrix}I&0\\-Y&I
    \end{pmatrix}
    =
    \begin{pmatrix}I-A&-tAF-B\\-Y-C&I-tCF-D
    \end{pmatrix}.       \label{eq:transversality-block-reduction}
  \end{equation}

  Divide the last $r$ columns by $t$ and let $t$ tend to infinity.
  The limit has the rank of

  \[
    \begin{pmatrix}I-A&-A\\-Y-C&-C
    \end{pmatrix},
  \]

  namely $r+\operatorname{rank} (C+YA)\ge d-1$, by a column subtraction and a
  Schur complement on $I_r$. Thus some $(d-1)$-minor in \eqref{eq:transversality-block-reduction}
  is a nonzero polynomial in $t$. At some rational $t$ the rank is
  at least $d-1$. Lemma~\ref{lem:isometry-rank-parity} then forces precisely \eqref{eq:transverse-residual-rank}.

  For determinant plus one, $\det(I-GT(X,Y))$ is consequently a
  nonzero polynomial. For determinant minus one, some fixed
  $(d-1)$-minor is nonzero as a polynomial. Entries have degree at
  most two in the free variables, so their respective total degrees
  are at most $2d$ and $2(d-1)$. The Schwartz--Zippel bound
  says that a nonzero polynomial of total degree $D$ vanishes on at
  most a $D/L$ fraction of such a product grid
  \cite[Lemma~1 and Corollary~1]{schwartz}; see also Zippel~\cite{zippel}.
  This proves the stated probability estimate.

\end{proof}

\subsubsection{A signed residual from isometry differences}

Let $\mathcal H$ be the group of block-diagonal signed permutations
$\operatorname{diag}(H_+,H_-)$ of block order $r$.

\begin{lemma}[Rank correlation]

  \label{lem:isometry-residual-correlation}

  Let $G$ have any finite law on integral determinant-one $J_0$-isometries.
  Independently choose $H_L,H_R$ uniformly in $\mathcal H$ and the free
  entries of alternating $X,Y$ uniformly in $\{0,\ldots,L_0-1\}$.
  Put

  \[
    U=\tfrac12 H_LPGT(X,Y)PH_R=N_U/2.
  \]

  For independent copies $U,V$, push the signed weight $-\det U\det V$
  to $E=2(U-V)$. The resulting integer measure $\omega$ has mass zero,
  variation one, and

  \[
    \langle g_d,\omega\rangle\ge1-2d/L_0.
  \]

  In particular $L_0=80d$ gives correlation at least $39/40$.

\end{lemma}
\begin{proof}

  Since $P^2=2I$, both $U,V$ preserve the same form $J$, and $N_U,N_V$
  are integral. By Lemma~\ref{lem:isometry-rank-parity},
  $\det U\det V=(-1)^{\operatorname{rank}(U-V)}$.
  Thus every contribution at one residual has the same sign,
  $(-1)^{d-1+\operatorname{rank}E}$, proving variation one.
  The determinant of a uniform compact factor is balanced, so the
  signed mass is zero.

  Condition on one copy and on every factor of the other except its
  two shears. Multiplying by fixed invertible matrices and cyclically
  moving factors reduces its residual rank to
  $\operatorname{rank}(I-G_0T(X,Y))$ for a fixed $J_0$-isometry $G_0$.
  Theorem~\ref{thm:two-shear-transversality} bounds the probability of failing
  to attain the largest rank of the appropriate parity by $2d/L_0$.
  The relative determinant component is balanced. Therefore the
  probability $p_0$ of full rank is at least $(1-2d/L_0)/2$.
  Every full-rank residual has negative weight, and
  $g_d=1-2\mathbf1_{\mathrm{full\ rank}}$. Mass zero now gives
  $\langle g_d,\omega\rangle=2p_0\ge1-2d/L_0$.

\end{proof}

\subsection{Positive correction in Fourier algebra norm}\label{sec:positive-correction}

For a finite group $G$, with probability-normalized Haar inner
products, use

\begin{equation}
  \|f\|_{\mathcal{A}}=\sum_{\pi\in\widehat G}
  d_\pi\left\|\widehat f(\pi)\right\|_{\mathrm{tr}} ,\qquad
  \widehat f(\pi)=\mathbb{E} _G f(g)\pi(g)^* .                   \label{eq:fourier-algebra-norm}
\end{equation}

This finite-group formula follows from Eymard's coefficient-norm
characterization~\cite[Lemma~2.14]{eymard} and the equality
$A(G)=B(G)$ for compact groups~\cite[Example~3.6(1)]{eymard}.
To verify the normalization, write a unitary coefficient as
$f(g)=\eta^*\rho(g)\xi$ and decompose $\rho$ into irreducibles.
Schur orthogonality gives
\[
  d_\pi\widehat f(\pi)=\sum_j\xi_{\pi,j}\eta_{\pi,j}^*,
  \qquad
  \|f\|_{\mathcal A}
  \le\sum_{\pi,j}\|\xi_{\pi,j}\|\,\|\eta_{\pi,j}\|
  \le\|\xi\|\,\|\eta\|.
\]
Conversely, Fourier inversion and a singular-value decomposition of
each $\widehat f(\pi)$ give such a coefficient with
$\|\xi\|\,\|\eta\|=\sum_\pi d_\pi\|\widehat f(\pi)\|_{\mathrm{tr}}$.
Taking the infimum over coefficients proves~\eqref{eq:fourier-algebra-norm}.
This norm satisfies

\begin{equation}
  \|f\|_\infty\le\|f\|_{\mathcal{A}},\qquad
  \|fg\|_{\mathcal{A}}\le\|f\|_{\mathcal{A}}\|g\|_{\mathcal{A}}. \label{eq:fourier-algebra-inequalities}
\end{equation}

These are the supremum bound in Eymard's Remark~2.15(1) and the
Banach-algebra bound in Proposition~2.16, respectively~\cite{eymard}.
Projection onto any sum of complete irreducible coefficient
spaces is also a contraction, since it drops terms from \eqref{eq:fourier-algebra-norm}.

The next lemma uses the positive-reweighting method of
Sherstov~\cite[Theorem~4.1]{sherstov-correction}: solve a near-identity
Gram system to make small correlations exactly zero while retaining
positive weights. His theorem is stated for Boolean functions and
bounds on sums of pairwise correlations. Here we use complete
irreducible coefficient spaces and control the correction in the
Fourier algebra norm, which gives the uniform bound needed below.
\begin{lemma}[Exact positive correction]

  \label{lem:positive-correction}

  Let $\mathcal{W}$ be a mean-zero sum of
  complete irreducible coefficient spaces. Let $\mathcal{V}$ be
  another such sum including constants, $\mathcal{W}$, and all
  products of two elements of $\mathcal{W}$.
  Assume both spaces are closed under complex conjugation.
  Let $\mu$ be a probability measure with Haar density $h$.
  If

  \begin{equation}
    e=P_{\mathcal{V}}h-1,\qquad
    \|e\|_{\mathcal{A}}\le\eta\le1/4,     \label{eq:projected-density-error}
  \end{equation}

  then there is a real $a\in\mathcal{W}$ with
  $\|a\|_\infty\le\eta/(1-\eta)\le1/3$ such that

  \begin{equation}
    \sigma=\frac{(1-a)\mu}{\mathbb{E} _\mu(1-a)}
    \quad\hbox{has}\quad
    \mathbb{E} _\sigma f=0\quad(f\in\mathcal{W}).      \label{eq:corrected-probability-law}
  \end{equation}

  If $\mu$ and the defining coefficient functions are rational,
  the correcting weights can be rational.

\end{lemma}
\begin{proof}

  For $a\in\mathcal{W}$, products $a\overline b$ with
  $b\in\mathcal{W}$ lie in $\mathcal{V}$. Orthogonality therefore gives

  \[
    P_{\mathcal{W}}(ha)=a+P_{\mathcal{W}}(ea).
  \]

  On $\mathcal{W}$ the operator
  $a\mapsto P_{\mathcal{W}}(ea)$ has norm at most $\eta$
  in $\|\cdot\|_{\mathcal{A}}$. The Neumann series solves

  \[
    a+P_{\mathcal{W}}(ea)=P_{\mathcal{W}}e
  \]

  with norm at most $\eta/(1-\eta)$.
  Closure under complex conjugation and real $h$ imply the solution is real.
  It satisfies $P_{\mathcal{W}}(h(1-a))=0$, so \eqref{eq:corrected-probability-law} follows,
  with weights in $[2/3,4/3]$ before normalization.
  In a rational basis the equations are an invertible rational
  Gram system, proving the final claim.

\end{proof}

\subsection{Mixing on frame orbits over odd finite fields}\label{sec:finite-field-frame-mixing}

Let $q$ be an odd prime power. Fix $1\le k\le r/2$ and a
symmetric $k$-by-$k$ matrix $C$ over $\mathbb{F} _q$. Write

\begin{equation}
  \mathcal{X}_C=\left\{\binom AB:
    A,B\in\mathbb{F} _q^{r\times k},\
  \operatorname{rank} \binom AB=k,\ A^tB+B^tA=C\right\}.               \label{eq:frame-gram-fiber}
\end{equation}

This set is nonempty: take $B=\binom{I_k}{0}$ and
$A=\binom{C/2}{0}$. Let $J_C$ be the orthogonal projection onto constant
functions on $\mathcal X_C$, given by uniform averaging.
Upper and lower alternating shears act by

\begin{equation}
  U(X)\binom AB=\binom{A+XB}{B},\qquad
  L(Y)\binom AB=\binom A{B+YA},\qquad X^t=-X,\ Y^t=-Y.   \label{eq:frame-shear-actions}
\end{equation}

Their uniform averages $\mathcal{U} ,\mathcal{L}$ are self-adjoint projections
for counting-measure inner products. Define

\begin{equation}
  v_k=\frac{q^k-1}{q-1},\qquad
  \delta_{r,k}(q)
  =v_kq^{-r}(1+2q+q^2+q^k).                          \label{eq:frame-mixing-error}
\end{equation}

Total variation distance between probability measures includes
the factor $1/2$.

\begin{theorem}[Uniform mixing on a frame fiber]

  \label{thm:finite-field-frame-mixing}

  For every $x\in\mathcal{X}_C$,

  \begin{equation}
    d_{\mathrm{TV}}\bigl((\mathcal{U} \mathcal{L} )^2\delta_x,J_C\delta_x\bigr)
    \le\delta_{r,k}(q).         \label{eq:frame-mixing-total-variation}
  \end{equation}

  The same bound holds with $\mathcal{U} ,\mathcal{L}$ interchanged. Therefore

  \begin{equation}
    \|(\mathcal{U} \mathcal{L} )^2-J_C\|
    \le2\delta_{r,k}(q).       \label{eq:frame-mixing-operator-norm}
  \end{equation}

  For $r\ge8k$, the shear groups act transitively on $\mathcal X_C$.

\end{theorem}
\begin{proof}

  For a nonzero vector $b$, the map $X\mapsto Xb$ from alternating
  matrices is onto $b^\perp$, with equal-sized fibers. Thus, for any
  full frame $\binom AB$ and any nonzero $w\in\mathbb{F} _q^k$,
  $Aw+XBw=0$ has probability at most $q^{-(r-1)}$ if $Bw\ne0$.
  If $Bw=0$, it is impossible, since then $Aw\ne0$.
  Taking one representative of each projective line of $w$ gives

  \begin{equation}
    \Pr_X[\operatorname{rank} (A+XB)<k]\le\alpha_1:=v_kq^{1-r}.           \label{eq:single-shear-rank-failure}
  \end{equation}

  The analogous bound holds for $B+YA$. It is uniform in the
  starting frame and does not assume either block is full rank.

  If $B$ has rank $k$, its upper orbit consists of all solutions
  of $A^tB+B^tA=C$. Each has equal probability, and the number
  of solutions is

  \begin{equation}
    q^h,\qquad h=rk-k(k+1)/2.                \label{eq:frame-orbit-size}
  \end{equation}

  Indeed a change $B\mapsto SB$, $A\mapsto S^{-t}A$ can put
  $B$ into $\binom{I_k}{0}$, where the assertion follows by
  prescribing the symmetric part of the first $k$ rows of $A$.
  The orbit assertion follows by prescribing the corresponding
  columns of the alternating matrix in \eqref{eq:frame-shear-actions}.

  Fix now full-rank $B,B'$ such that $\operatorname{rank} [B\ B']=2k$.
  The two simultaneous constraints

  \begin{equation}
    A^tB+B^tA=C,\qquad A^tB'+(B')^tA=C             \label{eq:simultaneous-gram-constraints}
  \end{equation}

  have exactly $q^{rk-k(k+1)}$ solutions. Change coordinates
  so that $B$ has the first $k$ standard columns and $B'$ has
  the next $k$. The constraints prescribe the symmetric parts
  of two disjoint $k$-by-$k$
  blocks of $A$.
  Among these solutions the fraction with $\operatorname{rank}  A<k$ is at most

  \begin{equation}
    \alpha_2:=v_kq^{2-r}.         \label{eq:double-constraint-rank-failure}
  \end{equation}

  To prove this, fix a projective direction $w$. Make a common
  invertible column change sending the first coordinate vector
  to $w$, followed by the row change just described.
  The first column of $A$ has exactly $r-2$ independent uniform
  coordinates; its two remaining coordinates are prescribed by
  the two diagonal entries in \eqref{eq:simultaneous-gram-constraints}. Thus $Aw=0$ is either impossible
  or has probability $q^{-(r-2)}$. A union bound proves \eqref{eq:double-constraint-rank-failure}.

  Start with full-rank $B$ and average first over the upper orbit,
  then over a lower orbit. The resulting marginal probability
  of any $B'$ with $\operatorname{rank} [B\ B']=2k$ is at least

  \begin{equation}
    q^{-2h}(1-\alpha_2)q^{rk-k(k+1)}
    =(1-\alpha_2)q^{-rk},      \label{eq:frame-transition-lower-bound}
  \end{equation}

  when $\alpha_2\le1$. If $\alpha_2>1$, the later bound is
  trivial and needs no lower estimate. In \eqref{eq:frame-transition-lower-bound} we counted only
  the intermediate matrices $A$ of full rank, for which the
  lower orbit has size $q^h$; all other contributions are nonnegative.

  For a uniformly random $r$-by-$k$ matrix $B'$, a projective
  union bound gives

  \begin{equation}
    \Pr[\operatorname{rank} [B\ B']<2k]\le\beta:=v_kq^{k-r},\qquad
    \Pr[\operatorname{rank}  B'<k]\le\alpha_0:=v_kq^{-r}.                \label{eq:random-frame-rank-failure}
  \end{equation}

  For the first inequality, $B'w$ must belong to the fixed
  $k$-dimensional column space of $B$ for some nonzero $w$.
  Summing the common mass in \eqref{eq:frame-transition-lower-bound} shows that the marginal
  of the output $B'$ is within $\alpha_2+\beta$ of the uniform law on all matrices,
  and hence within $\alpha_2+\beta+\alpha_0$ of the uniform law
  on full-rank matrices.

  The chronological order of $(\mathcal{U} \mathcal{L} )^2$ is lower, upper,
  lower, upper. Its first step makes $B$ full rank except with
  probability $\alpha_1$. Its middle two steps have the preceding
  marginal bound. On full-rank $B'$ the last step makes $A'$
  exactly uniform subject to its Gram constraint. Thus the whole
  output is within
  $\alpha_1+\alpha_2+\beta+\alpha_0$ of the uniform law on
  the part of $\mathcal{X}_C$ with $B'$ full rank.

  The uniform law on $\mathcal{X}_C$ is stationary under the lower
  average. By \eqref{eq:single-shear-rank-failure}, its probability of a rank-deficient $B$ is
  at most $\alpha_1$. Conditional on full-rank $B$, it is precisely
  the comparison law above, by \eqref{eq:frame-orbit-size}. The triangle inequality gives
  \eqref{eq:frame-mixing-total-variation}, since
  $2\alpha_1+\alpha_2+\beta+\alpha_0=\delta_{r,k}(q)$.
  Interchanging $A,B$ proves the adjoint statement. The row and
  column absolute-sum bounds for $(\mathcal{U} \mathcal{L} )^2-J_C$ are both
  $2\delta_{r,k}(q)$, proving \eqref{eq:frame-mixing-operator-norm} by the Schur bound.

  Finally, if $r\ge8k$, then
  $v_k\le(3/2)q^{k-1}$ and $1+2q+q^2+q^k\le5q^{k+1}$ give
  $2\delta_{r,k}(q)\le15q^{2k-r}<1$.
  If there were two shear-group orbits, an orbit indicator minus its
  mean would be a nonconstant fixed vector, contradicting \eqref{eq:frame-mixing-operator-norm}.
  This proves transitivity without any nondegeneracy assumption on $C$.

\end{proof}

\subsection{Congruence orbits of primitive frames}\label{sec:primitive-frame-congruence}

Let $q=p^a$ with $p$ odd. Let $J_0=\left(
  \begin{smallmatrix}0&I_r\\I_r&0
\end{smallmatrix}\right)$,
and let $\mathcal K_q$ be the kernel of reduction from
$\operatorname{SO}(J_0,\mathbb Z/q)$ to $\operatorname{SO}(J_0,\mathbb F_p)$.

The following is a congruence refinement of Witt extension.
The usual extension theorem over semiperfect rings extends an isometry
between direct summands of a unimodular quadratic space
\cite[Corollary~4.9]{first}. We additionally require the extension to
be the identity modulo $p$, and prove this by successive corrections.

\begin{lemma}[Congruence orbits on primitive frames]

  \label{lem:primitive-frame-congruence}

  For odd $p$ and $k\le d$, the group $\mathcal{K}_q$ is transitive
  on the $d\times k$ frames of column rank $k$ modulo $p$ with
  a prescribed reduction modulo $p$ and a prescribed Gram
  matrix $X^TJ_0X$ modulo $q$.

\end{lemma}
\begin{proof}

  Suppose two frames have already been matched modulo $p^s$.
  Write their first-order difference as $Z$ modulo $p$ and
  their common field reduction as $X$. Equality of Gram matrices
  implies $X^TJ_0Z+Z^TJ_0X=0$. Choose a left inverse $V$ of $X$
  and put $T=J_0Z$. Then

  \[
    W=TV-V^TT^T+V^T(T^TX)V
  \]

  is skew-symmetric and satisfies $WX=T$. Thus $A=J_0^{-1}W$
  is an orthogonal Lie algebra element with $AX=Z$.
  Lift $W$ skew-symmetrically over the local ring and apply the
  Cayley map $(I-p^sA/2)^{-1}(I+p^sA/2)$. Its first-order term is $I+p^sA$, so it matches the frames
  modulo $p^{s+1}$ with an element of $\mathcal{K}_q$.
  Iteration proves transitivity. The full column rank supplies
  the left inverse; the Gram matrix need not be nondegenerate.

\end{proof}

\subsection{Finite-module designs and structured residuals}\label{sec:finite-module-designs}

\subsubsection{Representation levels on finite modules}

Fix an odd prime $p$, an exponent $a\ge1$, and a positive integer
$b$. Let $\Gamma_a$ be the group generated, modulo $p^a$, by

\[
  U(X)=
  \begin{pmatrix}I_r&X\\0&I_r
  \end{pmatrix},\qquad
  L(Y)=
  \begin{pmatrix}I_r&0\\Y&I_r
  \end{pmatrix},\qquad
  X^t=-X,\quad Y^t=-Y.
\]

It preserves $J_0=\left(
  \begin{smallmatrix}0&I_r\\I_r&0
\end{smallmatrix}\right)$.
For every finite abelian $p$-group $M$ of exponent dividing
$p^a$ and with at most $b$ generators, let

\[
  \mathcal{H}_M=\ell^2(M^{2r}).
\]

The matrices act by integer linear combinations of the $2r$
coordinates. Write $|M|=p^\ell$ and $k=\dim_{\mathbb{F} _p}M/pM$.
The level of a type occurring in these representations is the
least $\ell$ for which it occurs in some $\mathcal{H}_M$ in
this family. The zero group has level zero and gives just the
trivial type. Let $N_M$ project onto the complete isotypic
components in $\mathcal{H}_M$ whose types have level $\ell$.

Let $\mathsf{U} ,\mathsf{L}$ be the uniform averages of the two shear groups.

\begin{theorem}[Contraction by module size]

  \label{thm:module-size-contraction}

  Suppose $r\ge80b$. For every irreducible type of level
  $\ell\ge1$,

  \begin{equation}
    \|(\mathsf{U} \mathsf{L} )^2\|\le p^{-3r\ell/4}.       \label{eq:module-size-contraction}
  \end{equation}

  The bound is independent of the prime-power exponent $a$.

\end{theorem}
\subsubsection{The symplectic Fourier model}

Let $M^*=\operatorname{Hom}(M,\mathbb{Q}/\mathbb{Z} )$, and put
$\mathcal{A}=M^*\oplus M$, with the perfect alternating pairing

\[
  \omega((\xi,x),(\eta,y))=\xi(y)-\eta(x)\in\mathbb{Q}/\mathbb{Z} .
\]

Partial Fourier transformation of the first $r$ coordinates
identifies $\mathcal{H}_M$ with $\ell^2(\mathcal{A}^r)$.
In this model the block swap becomes the self-adjoint unitary

\begin{equation}
  (\mathsf{S}  f)(z)=|M|^{-r}\sum_{w\in\mathcal{A}^r}
  \exp\!\left(2\pi i\sum_i\omega(z_i,w_i)\right)f(w),
  \qquad \mathsf{S} ^2=I.                                      \label{eq:symplectic-fourier-transform}
\end{equation}

An upper shear acts by a phase depending only on the pairings
$\omega(z_i,z_j)$. Thus

\begin{equation}
  \mathsf{U} =P_{\mathcal{C}},\qquad
  \mathcal{C}=\{z:\langle z_1,\ldots,z_r\rangle
  \text{ is isotropic in }\mathcal{A}\},
  \qquad \mathsf{L} =\mathsf{S} \mathsf{U} \mathsf{S} .                               \label{eq:shear-isotropic-projection}
\end{equation}

The same formulas follow directly from character orthogonality
as over a field; the characters now take values of order dividing
$p^a$. The swap normalizes $\Gamma_a$, and its action on every
$\mathcal{H}_{M'}$ preserves the collection of lower-level types.
Consequently $N_M$ commutes with $\mathsf{S} ,\mathsf{U} ,\mathsf{L}$.

\begin{lemma}[Smaller module copies]

  \label{lem:smaller-module-representations}

  For any isotropic subgroup $I\le\mathcal{A}$, the space

  \begin{equation}
    \mathcal{K}_I=\{f:\operatorname{supp}f\subseteq(I^\perp)^r,
    \ f(z+t)=f(z)\text{ for every }t\in I^r\}
    \label{eq:isotropic-quotient-subspace}
  \end{equation}

  is equivalent to $\mathcal{H}_{M'}$ for a finite abelian
  $p$-group $M'$ satisfying

  \begin{equation}
    |M'|=|M|/|I|,\qquad \exp M'\mid p^a,\qquad
    \dim_{\mathbb{F} _p}M'/pM'\le k.              \label{eq:quotient-module-parameters}
  \end{equation}

  In particular $N_M\mathcal{K}_I=0$ when $I\ne0$.

\end{lemma}
\begin{proof}

  The induced alternating pairing on $I^\perp/I$ is perfect.
  By the standard classification of finite symplectic abelian groups
  \cite[Corollary~5.7]{prasad-shapiro-vemuri}, $I^\perp/I$ is isometric
  to $(M')^*\oplus M'$ for a finite abelian $p$-group $M'$.

  Its size is $|M|^2/|I|^2$, and its exponent divides $p^a$.
  Every subgroup of $\mathcal{A}$ has at most $2k$ generators:
  its $p$-torsion embeds in $\mathcal{A}[p]$, and the number of
  generators of a finite abelian $p$-group equals the dimension
  of its $p$-torsion. Quotients do not increase the number of
  generators. Hence $I^\perp/I$ has at most $2k$ generators,
  which gives \eqref{eq:quotient-module-parameters}.

  Inflate a function on $(I^\perp/I)^r$ to $(I^\perp)^r$,
  multiply it by $|I|^{-r/2}$, and set it to zero elsewhere.
  This is a unitary identification with \eqref{eq:isotropic-quotient-subspace}. Upper-shear phases
  descend to the quotient, since all rows annihilate $I$.
  In \eqref{eq:symplectic-fourier-transform}, summing over $I^r$ gives zero off $(I^\perp)^r$;
  on that set it gives a factor $|I|^r$. The quotient Fourier
  normalization is therefore $|M'|^{-r}$. The identification
  intertwines both upper shears and $\mathsf{S}$, hence lower shears.
  This proves the representation equivalence.

\end{proof}

\subsubsection{A subgroup expansion of the isotropic projection}

For a subgroup $W\le\mathcal{A}$, let $P_W$ be multiplication
by $1_{W^r}$. Let $\mu(W,V)$ be the Möbius function of the
subgroup lattice, and for isotropic $W$ put

\begin{equation}
  c(W)=\sum_{\substack{V\ge W\\V\text{ isotropic}}}
  \mu(W,V).       \label{eq:isotropic-mobius-coefficients}
\end{equation}

\begin{lemma}

  \label{lem:isotropic-projection-expansion}

  \begin{equation}
    P_{\mathcal{C}}=\sum_{W\text{ isotropic}}c(W)P_W,
    \qquad |c(W)|\le p^{6k^2},
    \qquad \#\{W\le\mathcal{A}\}\le p^{4k\ell}.           \label{eq:isotropic-projection-expansion}
  \end{equation}

\end{lemma}
\begin{proof}

  For an isotropic subgroup $R$, every interval $[R,V]$ with
  isotropic $V$ is the full subgroup interval. Möbius inversion
  in \eqref{eq:isotropic-mobius-coefficients} gives $\sum_{W\ge R}c(W)=1$. If $R$ is nonisotropic, no
  isotropic subgroup contains $R$. Apply these facts to the subgroup generated
  by the rows to obtain the projection identity.

  For a finite abelian $p$-group $T$, the classical subgroup-lattice
  formula~\cite[Corollary~3.5]{hawkes-isaacs-ozaydin} gives

  \begin{equation}
    \mu(0,T)=
    \begin{cases}
      (-1)^j p^{j(j-1)/2},&T\simeq\mathbb{F} _p^j,\\
      0,&T\text{ is not elementary abelian}.
    \end{cases}                                          \label{eq:subgroup-mobius-formula}
  \end{equation}

  Only elementary subgroups $V/W$ of $(\mathcal{A}/W)[p]$ contribute
  to \eqref{eq:isotropic-mobius-coefficients}. The latter space has dimension at most $2k$. It has
  at most $p^{4k^2}$ subspaces, by counting generating $2k$-tuples,
  and each value in \eqref{eq:subgroup-mobius-formula} has magnitude at most $p^{2k^2}$.
  This proves the coefficient bound. Finally every subgroup of
  $\mathcal{A}$ has at most $2k$ generators, so counting their
  generating tuples gives $|\mathcal{A}|^{2k}=p^{4k\ell}$.

\end{proof}

\begin{lemma}[Fourier blocks]

  \label{lem:symplectic-fourier-blocks}

  For isotropic $W,W'$, put $I=W\cap(W')^\perp$. Then

  \begin{equation}
    \|P_W\mathsf{S}  P_{W'}\|=|M|^{-r}(|W'||I|)^{r/2}.         \label{eq:symplectic-fourier-block-norm}
  \end{equation}

  After compression by $N_M$, this block is zero unless the pairing
  between $W$ and $W'$ is perfect. A surviving block has
  $|W|=|W'|\le|M|$ and norm at most $|M|^{-r/2}$.

\end{lemma}
\begin{proof}

  For one row, the unnormalized row Gram matrix has entry $|W'|$
  when two rows differ by $I$, and zero otherwise. Its nonzero
  eigenvalues are $|W'||I|$. Tensor over $r$ rows and include
  the normalization in \eqref{eq:symplectic-fourier-transform}, proving \eqref{eq:symplectic-fourier-block-norm}.

  The range is supported on $W^r$ and invariant under translation
  by $I^r$. Since $W$ is isotropic, $W\subseteq I^\perp$.
  Thus the range lies in $\mathcal{K}_I$. Lemma~\ref{lem:smaller-module-representations}
  kills it when $I\ne0$; the adjoint gives the right-kernel
  condition. If both kernels vanish, the maps into the dual groups
  are injective in both directions, so the pairing is perfect.
  Finally $W\subseteq W^\perp$ implies $|W|^2\le|\mathcal{A}|$.

\end{proof}

\begin{proof}[Proof of Theorem~\ref{thm:module-size-contraction}]

  Choose $M$ of minimum size for the type, with $k\le b$.
  Put $A=\mathsf{U} \mathsf{S} \mathsf{U}$ and expand both indicators by \eqref{eq:isotropic-projection-expansion}.
  Lemma~\ref{lem:symplectic-fourier-blocks} and the triangle inequality give

  \begin{equation}
    \begin{split}
      \|A|_{N_M\mathcal{H}_M}\|
      &\le p^{8k\ell+12k^2-r\ell/2}\\
      &\le p^{20k\ell-r\ell/2}
      \le p^{-r\ell/4},
    \end{split} \label{eq:first-occurrence-compression-bound}
  \end{equation}

  since $\ell\ge k$ and $r\ge80k$.
  Since $\mathsf L=\mathsf S\mathsf U\mathsf S$ and $\mathsf U^2=\mathsf U$,
  we have $(\mathsf{U} \mathsf{L} )^2=A^3\mathsf{S}$, and all operators preserve
  the sum $N_M\mathcal H_M$ of the complete isotypic components of level $\ell$.
  Cubing \eqref{eq:first-occurrence-compression-bound} proves \eqref{eq:module-size-contraction}.

\end{proof}

\subsubsection{Exact averaging at small odd moduli}

Let $\Gamma_m$ be the group generated by the integer shears
modulo $m$, and let $J_m$ be its Haar orbit average.
Here $\rho_m$ denotes the permutation representation on the functions
of an $s$-column input, with action
$[\rho_m(g)f](X)=f(g^{-1}X)$.
For an integer matrix $G$, write
$\|G\|_{\mathrm{row}}=\max_i\sum_j|G_{ij}|$.

\begin{theorem}

  \label{thm:finite-module-design}

  For integers $B\ge3$, $s\ge1$, and $r\ge160s$, there is a
  finite positive rational law on integral determinant-one
  $J_0$-isometries $G$ such that

  \begin{equation}
    \mathbb{E} \rho_m(G)=J_m\quad\text{on }(\mathbb{Z} /m)^{2r\times s}
    \quad(m\le B\text{ odd}).          \label{eq:exact-modular-orbit-average}
  \end{equation}

  Every support element is a product of $32$ alternating shears,
  and

  \begin{equation}
    \|G\|_{\mathrm{row}}\le(1+rB^4)^{32}. \label{eq:shear-word-height}
  \end{equation}

\end{theorem}
\begin{proof}

  Let $D$ be the odd part of $\operatorname{lcm}(1,\ldots,B)$.
  CRT identifies $\Gamma_D$ with $\prod_{p^a\parallel D}\Gamma_a$:
  each local root parameter can be prescribed independently.
  Sample eight independent blocks $ULUL$. Every free integer
  root coordinate is an independent copy of

  \begin{equation}
    Z=\sum_{j=1}^{B^2}Z_j,
    \qquad Z_j\sim\operatorname{Unif}\{0,\ldots,j-1\}.
    \label{eq:shear-parameter-law}
  \end{equation}

  This is uniform modulo every integer at most $B^2$, and
  $0\le Z<B^4/2$. Let $h$ be the Haar density of the reduced
  word law on $\Gamma_D$.

  Let $\mathcal{W}$ be the sum of the mean-zero coefficient spaces
  of the $s$-column actions at all odd $m\le B$. Let $\mathcal{V}$
  be the sum of the $2s$-column coefficient spaces at divisors
  $q\mid D$ with $q\le B^2$, including constants. These are
  complete irreducible coefficient spaces closed under complex
  conjugation, and $\mathcal{W}\mathcal{W}\subseteq\mathcal{V}$:
  use $\operatorname{lcm}(m,n)$ and zero columns at the missing
  CRT factors.

  At a fixed prime, use first levels in the family with at most
  $b=2s$ generators. There are at most $P(\ell)\le2^{\ell-1}$
  isomorphism types of groups of size $p^\ell$: their invariant
  factor types are partitions of $\ell$, which inject into its
  ordered compositions. Therefore the sum of squared dimensions
  of irreducible types first appearing at level $\ell$ is at most

  \begin{equation}
    P(\ell)p^{4r\ell}.              \label{eq:representation-dimension-sum}
  \end{equation}

  The coefficient space of each $\mathcal{H}_M$ has dimension
  at most $(\dim\mathcal{H}_M)^2=|M|^{4r}$.
  By \eqref{eq:module-size-contraction}, the eight-block operator norm on such a type is at most
  $p^{-6r\ell}$.

  Every type included in $\mathcal{V}$ factors through a conductor
  $q\le B^2$. At that conductor the parameters in \eqref{eq:shear-parameter-law} are
  uniform under CRT, so its averaged representation is the tensor
  product of the local averages. In Fourier algebra norm this gives

  \begin{equation}
    \begin{split}
      \|P_{\mathcal{V}}h-1\|_{\mathcal{A}}
      &\le\prod_{3\le p\le B}
      \left(1+\sum_{\ell\ge1}2^{\ell-1}p^{-2r\ell}\right)-1
      <\frac14.
    \end{split} \label{eq:design-density-error}
  \end{equation}

  The sum in one factor is $p^{-2r}/(1-2p^{-2r})\le2p^{-2r}$.
  The sum over primes is at most $4\cdot3^{-2r}\le1/16$;
  thus the product minus one is at most $e^{1/16}-1<1/4$.
  The larger product bounds a sum of nonnegative errors. Independence on the entire group
  $\Gamma_D$ is not assumed.

  Apply Lemma~\ref{lem:positive-correction}: solve
  $a+P_{\mathcal{W}}((P_{\mathcal{V}}h-1)a)=P_{\mathcal{W}}(h-1)$
  in Fourier algebra norm and use the positive weight $1-a$.
  It lies in $[2/3,4/3]$ before normalization and has rational
  values in a rational coefficient basis. Lift the weight to
  the integer word by its reduction modulo $D$.
  The correction gives \eqref{eq:exact-modular-orbit-average} without changing the word support.
  Each shear has row norm at most $1+rB^4$, proving \eqref{eq:shear-word-height}.

\end{proof}

\subsubsection{Determinant-sign cancellation over odd residue rings}

The following orbit argument is separate from the design's
analytic estimate.

\begin{lemma}

  \label{lem:shear-congruence-kernel}

  For $r\ge3$ and odd $p$, the group $\Gamma_a$ contains
  the full kernel of

  \begin{equation}
    \operatorname{SO}(J_0,\mathbb{Z} /p^a)\longrightarrow
    \operatorname{SO}(J_0,\mathbb{F} _p).       \label{eq:orthogonal-congruence-reduction}
  \end{equation}

  If $r\ge8s$, every orthogonal isometry modulo $p^a$, including
  one of determinant $-1$, preserves each $\Gamma_a$-orbit on
  $s$-column inputs.

\end{lemma}
\begin{proof}

  For $1\le j<a$, the layer of the kernel from $p^{j+1}$ to
  $p^j$ is the additive Lie algebra

  \begin{equation}
    \left\{
      \begin{pmatrix}Z&B\\C&-Z^t
      \end{pmatrix}:
    B,C\text{ alternating}\right\}.
    \label{eq:orthogonal-congruence-layer}
  \end{equation}

  The shears with parameters divisible by $p^j$ supply the
  $B,C$ directions. Conjugate an upper such direction by
  $L(Y)$, with $Y$ alternating. Modulo the already supplied
  directions its $Z$ block is $-BY$.
  Products of alternating matrices span all $r$-by-$r$ matrices:
  for distinct $i,h,j$,

  \[
    (E_{ih}-E_{hi})(E_{hj}-E_{jh})=E_{ij},\qquad
    (E_{ij}-E_{ji})^2=-E_{ii}-E_{jj}.
  \]

  The second identity supplies every diagonal because $2$ is
  invertible and three distinct indices are available.
  Thus the image of $\Gamma_a$ in every layer \eqref{eq:orthogonal-congruence-layer} is surjective.
  Matching the layers successively proves \eqref{eq:orthogonal-congruence-reduction}.

  For a primitive $s$-frame $Y$, the finite-field orbit theorem
  in Section~\ref{sec:finite-field-frame-mixing} supplies an element
  of $\Gamma_1$ taking $\overline Y$ to $\overline{HY}$ for
  any isometry $H$: their Gram matrices agree. Lift it to
  $\Gamma_a$. The resulting frame and $HY$ have the same
  reduction and exact Gram matrix. The principal-kernel
  transitivity lemma in Section~\ref{sec:primitive-frame-congruence},
  together with \eqref{eq:orthogonal-congruence-reduction}, matches them exactly. For an arbitrary
  input, write $X=YTV$ by Smith normal form with $Y$ primitive,
  $T$ diagonal, and $V$ invertible; push this equality through
  right multiplication by $TV$. CRT gives the same conclusion
  over every odd composite modulus.

\end{proof}

Conjugate to $J=\operatorname{diag}(I_r,-I_r)$ coordinates by
$P=\left(
  \begin{smallmatrix}I&I\\I&-I
\end{smallmatrix}\right)$.
Follow $G$ in Theorem~\ref{thm:finite-module-design} by a uniform compact
signed permutation $H\in\mathcal H$ as in Section~\ref{sec:split-isometries}.
The rational isometry
$PGPH/2$ has clearing factor two. For each odd $m\le B$,
the Haar orbit average is invariant under every conjugated $H$
by Lemma~\ref{lem:shear-congruence-kernel}. Hence its determinant-signed
$s$-column operator is zero, since $\mathbb{E} _H\det H=0$.

\subsubsection{Residual height and Fourier zeros}

We now prove the residual statement used in Lemma~\ref{lem:residual-interface}.

\begin{corollary}

  \label{cor:structured-residual}

  Let $r$ be a power of two, $d=2r$, $r\ge160s$, $B\ge3$,
  and $L\ge1$. There is a finite rational signed residual
  $\nu$ on integer $d\times d$ matrices, of mass zero and variation
  at most one, with correlation
  $\langle g_d,\nu\rangle\ge403/440$ and entry height $K$ satisfying

  \begin{equation}
    \log_2(K+1)=O(\log d+\log B+\log(L+1)). \label{eq:structured-residual-height}
  \end{equation}

  Its variation on matrices of corank at least two is at most $1/40$.
  All Fourier derivatives through order $L$, including order zero,
  vanish at $Q/t$ whenever $t\le B$ and $Q$ factors through at
  most $s$ columns modulo $t$. In particular

  \begin{equation}
    |\widehat\nu(Q/t+\Delta)|\le2^{-(L+1)}
    \quad\text{if }\|\Delta\|_{\mathrm{entry,1}}\le(L+1)/(40K).
    \label{eq:structured-residual-neighborhood}
  \end{equation}

\end{corollary}
\begin{proof}

  Let $G$ have the positive integer-shear law of Theorem~\ref{thm:finite-module-design}.
  Independently choose the two-shear regularizer $T=T(X,Y)$ of width $80d$
  and compact signed permutations $H_L,H_R$. In the notation of
  Section~\ref{sec:split-isometries}, use the composition

  \[
    U=\frac12 H_LPGTPH_R=\frac{N_U}{2}.
  \]

  All factors in $N_U$ are integral. Since $P^2=2I$, this is a rational
  $J$-isometry with common clearing factor two. The regularizer lies
  between $G$ and the final $P$; the compact factors lie outside.

  For an odd modulus $m\le B$, averaging over $G$ on an $s$-column
  frame gives the orbit projector $J_m$. The regularizer belongs to
  $\Gamma_m$, so $J_m\rho_m(T)=J_m$. The compact right factor preserves
  each orbit by Lemma~\ref{lem:shear-congruence-kernel}; hence its
  balanced determinant sign makes the signed frame law zero.
  Conditioning on $H_L$ and pushing forward preserves this zero.

  For independent copies $U,V$, push the weight $-\det U\det V$ to the
  integer matrix $E=2(U-V)=N_U-N_V$. The proof of
  Lemma~\ref{lem:isometry-residual-correlation} applies: conditional
  on the other factors, cyclically moving the fixed invertible factors
  reduces the rank test to $I-G_0T(X,Y)$ for a fixed rational
  $J_0$-isometry $G_0$. Uniform transversality gives correlation at least
  $39/40$. Rank parity gives variation one, and the balanced compact
  signs give mass zero. Conditioning on either copy transfers its signed
  frame zero to $E$.
  Explicitly, for every fixed $X\in(\Z/m)^{d\times s}$, the signed
  distribution of $EX$ is zero. If $Q=AB$ factors through $s$ columns,
  then $\langle Q,E\rangle=\langle A,EB^T\rangle$, so the Fourier
  transform vanishes at $Q/m$. The same holds for every integer
  multiple of $Q$, since it has the same kind of factorization.

  The word bound \eqref{eq:shear-word-height}, the two-shear regularizer, and the two copies of
  $P$ use $O(\log d+\log B)$ entry bits.

  Dilate by $2^{\lfloor\log_2B\rfloor}$. At an even modulus $m=2^eu\le B$
  the dyadic component is zero and the odd component is the zero signed
  law; for $u=1$, use mass zero. CRT gives the joint zero.
  The added bit cost is $O(\log B)$.
  Finally apply the filter of Section~\ref{sec:modular-fourier-jets}
  with tail variation $1/32$ and order $L$. Its positive nonzero dilation
  nodes preserve rational ranks, and its correlation estimate gives
  $\langle g_d,\nu\rangle\ge403/440$. It adds
  $O(\log B+\log(L+1))$ entry bits, proving~\eqref{eq:structured-residual-height}.
  Its Taylor remainder proves \eqref{eq:structured-residual-neighborhood}.
  The transversality estimate also gives variation at most $1/40$ on
  coranks at least two before filtering. Every nonzero scalar dilation
  preserves ranks, and the filter's absolute coefficients sum to one
  after normalization. Hence the same exceptional-variation bound holds
  for the final residual.

\end{proof}

\subsection{Simultaneous modular Fourier jets}\label{sec:modular-fourier-jets}

\subsubsection{Scalar filters with small tail variation}

Throughout, $B,L\ge1$ are integers and $0<\epsilon<1$.
Define

\begin{equation}
  H=\left\lceil 2B^4L^2+\frac{3BL}{\epsilon}\right\rceil,
  \qquad N=(16L^2+1)H,\qquad I=\{H,\ldots,N\}.
  \label{eq:scalar-filter-nodes}
\end{equation}

\begin{theorem}

  \label{thm:modular-moment-filter}

  There are rational coefficients $c_s$, supported on $\{1\}\cup I$,
  such that

  \begin{equation}
    c_1=1,\qquad \tau:=\sum_{s\in I}|c_s|\le\epsilon,
    \label{eq:scalar-filter-tail-variation}
  \end{equation}

  and, simultaneously for every $1\le m\le B$, $a\in\mathbb{Z}/m$,
  and $1\le k\le L$,

  \begin{equation}
    \sum_{s\equiv a\bmod m}c_s s^k=0.     \label{eq:modular-scalar-moment-cancellation}
  \end{equation}

  In particular

  \begin{equation}
    N=O(B^4L^4+BL^3/\epsilon),\qquad
    \log N=O(\log(2B)+\log(2L)+\log\epsilon^{-1}).          \label{eq:scalar-filter-support-bound}
  \end{equation}

  No bound on the denominators of the coefficients is asserted.

\end{theorem}
\subsubsection{Discrete interpolation with polynomial and Fourier features}

Let $\mathcal{F}_B\subset\mathbb{R}/\mathbb{Z}$ be the set of distinct
rational frequencies of reduced denominator at most $B$, including
zero. Its size is $F=\sum_{m\le B}\varphi(m)\le B^2$,
with $\varphi(1)=1$. Distinct frequencies have cyclic separation
at least $B^{-2}$. Consider the space

\begin{equation}
  \mathcal{W}=\operatorname{span}_{\mathbb{C}}
  \{s^j\mathrm{e} ^{2\pi i\theta s}:\theta\in\mathcal{F}_B,\ 0\le j<L\}.
  \label{eq:polynomial-fourier-space}
\end{equation}

In particular it contains every $s^j1_{s\equiv a\bmod m}$
with $m\le B$ and $j<L$.

Put $T=|I|=16L^2H+1$ and $t_s=(2s-H-N)/(N-H)$.
For the Legendre polynomial $P_j$, use features

\begin{equation}
  \phi_{\theta,j}(s)=\sqrt{2j+1}\,P_j(t_s)\mathrm{e} ^{2\pi i\theta s},
  \qquad
  G=\frac1T\sum_{s\in I}\phi(s)\phi(s)^* .               \label{eq:interpolation-gram-matrix}
\end{equation}

There are $FL$ features; they span \eqref{eq:polynomial-fourier-space}.

\begin{lemma}

  \label{lem:interpolation-gram-bound}

  The Gram matrix and evaluation vector satisfy

  \begin{equation}
    \|G-\mathrm{Id}\|\le\frac{12B^4L^4}{T}<\frac12,
    \qquad \|\phi(1)\|^2\le \mathrm{e}  B^2L^2.                    \label{eq:interpolation-gram-bounds}
  \end{equation}

  Consequently there are real rational weights $u_s$, $s\in I$,
  such that

  \begin{equation}
    \sum_{s\in I}u_s f(s)=f(1)\quad(f\in\mathcal{W}),
    \qquad \sum_{s\in I}|u_s|<3BL.                         \label{eq:interpolation-reproducing-weights}
  \end{equation}

\end{lemma}
\begin{proof}

  We use the standard Legendre orthogonality and bounds
  \cite[Table~18.3.1 and Sections~18.9--18.10 and~18.14]{dlmf}:

  \begin{equation}
    \begin{gathered}
      \frac12\int_{-1}^{1}P_i(t)P_j(t)\,dt
      =\frac{1_{i=j}}{2j+1},\qquad
      |P_j(t)|\le1\quad(-1\le t\le1),\\
      |P'_j(t)|\le j(j+1)/2\quad(-1\le t\le1),\qquad
      |P_j(\cosh\alpha)|\le\mathrm{e} ^{j\alpha}\quad(\alpha\ge0).
      \label{eq:legendre-polynomial-bounds}
    \end{gathered}
  \end{equation}

  In particular, the Laplace integral representation
  \cite[equation~(18.10.5)]{dlmf}, continued as a polynomial identity, is

  \begin{equation}
    P_j(t)=\frac1\pi\int_0^\pi
    \bigl(t+\sqrt{t^2-1}\cos v\bigr)^j\,dv.
    \label{eq:legendre-integral-representation}
  \end{equation}

  Its integrand has absolute value at most one for
  $|t|\le1$, and at most $\mathrm{e} ^{j\alpha}$ at $t=\cosh\alpha$.
  For $j\ge1$, the derivative bound follows from
  $P'_j=\frac{j+1}{2}P_{j-1}^{(1,1)}$ and
  $|P_{j-1}^{(1,1)}(t)|\le j$ on $[-1,1]$
  \cite[equations~(18.9.15) and~(18.14.1)]{dlmf}; the case $j=0$
  is immediate.

  Write $p_j=\sqrt{2j+1}P_j$. For $0\le i,j<L$, the function
  $h=p_ip_j$ obeys

  \begin{equation}
    \|h\|_\infty\le2L,
    \operatorname{TV} _{[-1,1]}(h)\le4L^3.              \label{eq:legendre-product-variation}
  \end{equation}

  Indeed its derivative has supremum at most $2L^3$ by \eqref{eq:legendre-polynomial-bounds}.
  The average over the $T$ equally spaced points differs from
  $\frac12\int h$ by at most

  \begin{equation}
    \frac{\operatorname{TV} (h)}{T-1}+\frac{2\|h\|_\infty}{T}
    \le\frac{12L^3}{T}.                             \label{eq:legendre-quadrature-error}
  \end{equation}

  To see the first inequality, compare the integral with the
  left-endpoint average over $T-1$ intervals, then compare that
  average with the average including the last endpoint.

  For distinct frequencies, every consecutive partial sum of
  $\mathrm{e} ^{2\pi i(\theta-\theta')s}$ has magnitude at most
  $[2\|\theta-\theta'\|_{\mathbb{R}/\mathbb{Z}}]^{-1}\le B^2/2$.
  Abel summation therefore bounds the corresponding Gram entry by

  \begin{equation}
    \frac{B^2}{2T}\bigl(\|h\|_\infty+\operatorname{TV} (h)\bigr)
    \le\frac{3B^2L^3}{T}.         \label{eq:off-frequency-gram-bound}
  \end{equation}

  By \eqref{eq:legendre-quadrature-error}, \eqref{eq:off-frequency-gram-bound}, and orthogonality, every entry of $G-\mathrm{Id}$
  has absolute value at most $12B^2L^3/T$. Its operator norm is at
  most $FL$ times that quantity. Since $H\ge2B^4L^2$,
  this proves the first inequality in \eqref{eq:interpolation-gram-bounds}, in fact with bound $3/8$.

  Next, $t_1=-1-x$ where $0\le x=(H-1)/(8L^2H)<1/(8L^2)$.
  If $1+x=\cosh\alpha$, then $\alpha\le\sqrt{2x}<1/(2L)$.
  Parity and \eqref{eq:legendre-polynomial-bounds} imply $|P_j(t_1)|\le\mathrm{e} ^{1/2}$ for $j<L$.
  Thus

  \begin{equation}
    \|\phi(1)\|^2
    \le \mathrm{e}  F\sum_{j=0}^{L-1}(2j+1)
    =\mathrm{e}  FL^2\le\mathrm{e}  B^2L^2.
    \label{eq:interpolation-evaluation-bound}
  \end{equation}

  The reproducing weights are

  \begin{equation}
    u_s=\frac1T\phi(s)^*G^{-1}\phi(1).       \label{eq:minimum-norm-interpolation-weights}
  \end{equation}

  They satisfy $\sum_su_s\phi(s)=\phi(1)$ and are the unique
  minimum-Euclidean-norm solution of those constraints. Moreover

  \begin{equation}
    \sum_s|u_s|\le\sqrt{T\sum_s|u_s|^2}
    =\sqrt{\phi(1)^*G^{-1}\phi(1)}
    \le\sqrt{2\mathrm{e} }\,BL<3BL.                         \label{eq:interpolation-weight-variation}
  \end{equation}

  Although \eqref{eq:minimum-norm-interpolation-weights} uses complex features, $\mathcal W$ admits a basis over
  $\mathbb C$ consisting of rational-valued functions
  $s^j1_{s\equiv a\bmod m}$. Its evaluation matrix and right-hand
  side at $1$ are rational. The minimum-norm solution is therefore
  real and rational: if $R$ is the full-row-rank evaluation matrix
  in that basis, it is $R^t(RR^t)^{-1}f(1)$.

\end{proof}

\begin{proof}[Proof of Theorem~\ref{thm:modular-moment-filter}]

  Set $c_1=1$ and $c_s=-u_s/s$ on $I$. For each required
  residue class and degree, apply \eqref{eq:interpolation-reproducing-weights} to
  $f(s)=s^{k-1}1_{s\equiv a\bmod m}$. This gives \eqref{eq:modular-scalar-moment-cancellation}.
  Finally $\tau\le H^{-1}\sum_s|u_s|<3BL/H\le\epsilon$.
  The support estimate follows directly from \eqref{eq:scalar-filter-nodes}.

\end{proof}

\subsubsection{From scalar filters to matrix Fourier derivatives}

Let $\nu$ be a finite real signed measure on integer $d\times d$
matrices, with mass zero, variation at most one, and entry height
at most $K$. Let $D_s$ denote scalar dilation $E\mapsto sE$, and put

\begin{equation}
  \widetilde\nu=\frac1{1+\tau}\sum_s c_s(D_s)_*\nu,
  \quad \gamma=\frac{\sum_sc_s}{1+\tau}
  \ge\frac{1-\epsilon}{1+\epsilon},
  \quad K'=NK.                                         \label{eq:filtered-residual}
\end{equation}

\begin{theorem}

  \label{thm:matrix-fourier-jet-transfer}

  The measure $\widetilde\nu$ has mass zero, variation at most one,
  and entry height at most $K'$. For every $m\le B$, residue matrix
  $Z$, and multi-index $\alpha$ with $1\le|\alpha|\le L$,

  \begin{equation}
    \sum_{E'\equiv Z\bmod m}
    \widetilde\nu(E')(E')^\alpha=0.\label{eq:modular-matrix-moment-cancellation}
  \end{equation}

  Consequently every mixed Fourier derivative of positive order
  at most $L$ vanishes at every $\theta_0=Q/m$, $m\le B$.
  If the original transform vanishes at $s\theta_0$ for every
  positive integer $s$, the zeroth-order value vanishes there too.
  At such a zero and for any real matrix $\Delta$,

  \begin{equation}
    |\widehat{\widetilde\nu}(\theta_0+\Delta)|
    \le\frac{(2\pi K'\|\Delta\|_{\mathrm{entry,1}})^{L+1}}{(L+1)!}.
    \label{eq:filtered-residual-taylor-bound}
  \end{equation}

  In particular the bound is at most $2^{-(L+1)}$ when
  $\|\Delta\|_{\mathrm{entry,1}}\le(L+1)/(40K')$.

\end{theorem}
\begin{proof}

  For a fixed original atom $E$, the contribution to \eqref{eq:modular-matrix-moment-cancellation}, apart
  from $\nu(E)E^\alpha/(1+\tau)$, is

  \[
    \sum_{s:sE\equiv Z\bmod m}c_ss^{|\alpha|}.
  \]

  The condition on $s$ is a union of residue classes modulo $m$.
  Each class sum is zero by \eqref{eq:modular-scalar-moment-cancellation}, whether or not $E$ or $s$ is
  invertible modulo $m$. Summing over $E$ proves \eqref{eq:modular-matrix-moment-cancellation}.
  Weighting the residue identities by
  $\mathrm{e} ^{2\pi i\langle Q,Z\rangle/m}$ gives the derivatives.
  The constant Fourier value is
  $(1+\tau)^{-1}\sum_s c_s\widehat\nu(s\theta_0)$, proving
  its stated sufficient condition.

  For \eqref{eq:filtered-residual-taylor-bound}, expand $\mathrm{e} ^{2\pi i\langle\Delta,E'\rangle}$ through
  degree $L$ after inserting the phase at $\theta_0$.
  All these terms vanish. For real $t$, the Taylor remainder of
  $\mathrm{e} ^{it}$ has magnitude at most $|t|^{L+1}/(L+1)!$.
  Use $|\langle\Delta,E'\rangle|\le K'\|\Delta\|_{\mathrm{entry,1}}$
  and variation at most one. Finally $(L+1)!\ge((L+1)/\mathrm{e} )^{L+1}$
  and $2\pi\mathrm{e} <20$ give the neighborhood assertion.

\end{proof}

For example, if every signed $k$-frame projection of $\nu$ is
zero modulo every $m\le B$, then its Fourier transform is zero
at $Q/m$ whenever $Q$ has a factorization through $k$ columns
over $\mathbb{Z}/m$. All integer multiples of $Q$ have such a
factorization as well, so Theorem~\ref{thm:matrix-fourier-jet-transfer} applies.
Every pre-existing modular frame zero at other moduli is also
preserved under arbitrary integer scalar dilations; the new
derivative assertion only covers $m\le B$.

The constant-value hypothesis is essential. For the scalar signed
measure $\nu=(\delta_1-\delta_0)/2$, the value at $1/2$ of its
dilation mixture is $-(1+\tau)^{-1}\sum_{s\ {\mathrm{odd}}}c_s$.
That odd coefficient sum has magnitude at least $1-\tau>0$.
Thus canceling every positive-order jet does not create a
zeroth-order zero.

\subsubsection{Preserving singularity correlation under dilation}

Write $F(E)=+1$ for every rationally singular matrix and $-1$
for nonsingular matrices. Every node $s$ is positive and nonzero,
so it preserves every rational rank. Consequently

\begin{equation}
  \langle F,\widetilde\nu\rangle
  =\gamma\langle F,\nu\rangle.\label{eq:dilation-correlation-identity}
\end{equation}

For the residual in Lemma~\ref{lem:isometry-residual-correlation}, take
$L_0=80d$ and $\epsilon=1/32$. Its correlation after filtering is at least

\begin{equation}
  \langle g_d,\widetilde\nu\rangle
  \ge\frac{39}{40}\frac{31}{33}=\frac{403}{440}.
  \label{eq:filtered-residual-correlation}
\end{equation}

Every singular rank retains the positive label throughout.

\section{Bitwise-XOR singularity}\label{sec:xor-singularity}

This section proves Theorem~\ref{thm:xor-lower-bound} and a corollary for
Boolean entrywise products.

We first record the elementary randomized and Boolean quantum
consequences of published disjointness and symmetric-XOR theorems.
The dense-equation witness then supplies the full quantum bit-length
factor. Its constants and thresholds are independent of the relation
between dimension and bit length.
\subsection{Sum encodings and the reduction from AND to XOR}

For an $m\times m$ matrix $Z$ and vectors $a,b\in\Q^m$, define
\begin{equation}\label{eq:weighted-sum-matrix}
  T_{a,b}(Z)=
  \begin{pmatrix}
    0&a^{\mathsf T}&0\\0&I_m&Z\\b&0&I_m
  \end{pmatrix}.
\end{equation}
The bottom-right block has determinant one and inverse
$\bigl(
  \begin{smallmatrix}I&-Z\\0&I
\end{smallmatrix}\bigr)$.
Its Schur complement is $a^{\mathsf T}Zb$, giving
\begin{equation}\label{eq:weighted-sum-determinant-rank}
  \det T_{a,b}(Z)=a^{\mathsf T}Zb,\qquad
  \rank T_{a,b}(Z)=2m+\mathbf1[a^{\mathsf T}Zb\ne0].
\end{equation}
Let $\mathbf1_m$ denote the all-one column of length $m$, and write
$T(Z)=T_{\mathbf1_m,\mathbf1_m}(Z)$. Its determinant is
$\sum_{ij}Z_{ij}$, and its corank is at most one.
This representation of a sum as a determinant is the block-matrix
form of the graph construction in Ivanyos et al.\ \cite[Claim~2 in
the proof of Theorem~4]{ivanyos-et-al}; the vectors $a,b$ give its
weighted version.

For nonnegative integers, $u+v-(u\oplus v)=2(u\bitand v)$.
Consequently, for $d\times d$ matrices of $k$-bit integers,
\begin{equation}\label{eq:and-to-xor-rank}
  L(U,V)=
  \begin{pmatrix}U&U\oplus V&I\\0&V&I\\I&I&0
  \end{pmatrix}
  \quad\text{satisfies}\quad
  \rank L(U,V)=2d+\rank(U\bitand V),
\end{equation}
by~\eqref{eq:four-block}. Its determinant is
$(-2)^d\det(U\bitand V)$. This is a local XOR construction:
Alice and Bob respectively output
\begin{equation}\label{eq:and-to-xor-shares}
  \mathcal A(U)=
  \begin{pmatrix}U&U&I\\0&0&I\\I&I&0
  \end{pmatrix},
  \qquad
  \mathcal B(V)=
  \begin{pmatrix}0&V&0\\0&V&0\\0&0&0
  \end{pmatrix}.
\end{equation}
Their XOR is $L(U,V)$, every output entry has at most $k$ bits,
and the corank of $U\bitand V$ is preserved. In particular, this
also converts a Boolean entrywise product into a Boolean XOR matrix.

\subsection{Randomized lower bound from disjointness}\label{sec:randomized-xor-lower-bound}

We record the classical bound as an elementary consequence of the
published disjointness theorem~\cite{razborov} and the preceding
determinant encoding.

Take disjointness inputs $x,y\in\{0,1\}^{m^2k}$. Pack consecutive
blocks of $k$ bits into entries of $m\times m$ matrices $X,Y$.
Because all summands are nonnegative,
\[
  \sum_{i,j}(X_{ij}\bitand Y_{ij})=0
  \quad\Longleftrightarrow\quad x\cap y=\varnothing.
\]
The public entries of $T$ are Boolean, so
$T(X)\bitand T(Y)=T(X\bitand Y)$. Applying~\eqref{eq:and-to-xor-shares}
with $U=T(X)$ and $V=T(Y)$ gives an XOR matrix of order $6m+3$
and corank at most one, whose determinant is
\begin{equation}\label{eq:packed-disjointness-determinant}
  \det L(T(X),T(Y))=(-2)^{2m+1}
  \sum_{i,j}(X_{ij}\bitand Y_{ij}).
\end{equation}
Choose $m=\lfloor(n-3)/6\rfloor$ and pad with an identity block.
Razborov's classical disjointness theorem~\cite{razborov} gives
$\Omega(m^2k)=\Omega(n^2k)$ randomized communication.
Sending all of $A$ costs $n^2k$ bits, so the randomized complexity
is $\Theta(n^2k)$.

\subsection{Boolean quantum lower bound from symmetric XOR functions}
\label{sec:boolean-xor-lower-bound}

Let $N=m^2$ and $t=\lfloor N/2\rfloor$. Zhang--Shi's theorem on
symmetric XOR functions~\cite[Theorem~2]{zhang-shi-xor} gives
\begin{equation}\label{eq:exact-hamming-quantum-bound}
  Q^*_{1/3}\bigl[\,|x\oplus y|=t\,\bigr]=\Omega(N).
\end{equation}
For this exact-distance predicate, their parameter $r$ is $\Theta(N)$.
Their model permits unlimited prior entanglement. We give the local
encoding that transfers this published bound to Boolean singularity.

Index the bits as Boolean $m\times m$ matrices $X,Y$, and fix a
public Boolean matrix $H_t$ with exactly $t$ ones. For a Boolean matrix
$D$ of order $m$, put
\[
  Z(D)=\diag(D,H_t),\qquad
  a=(\mathbf1_m,-\mathbf1_m),\qquad b=\mathbf1_{2m},
  \qquad T_t(D)=T_{a,b}(Z(D)).
\]
By~\eqref{eq:weighted-sum-determinant-rank}, $T_t(D)$ has order $d=4m+1$,
determinant $\sum_{ij}D_{ij}-t$, and corank at most one.
Its negative entries are public. Write $T_t(D)=P(D)-Q$ using
its Boolean positive and negative parts. The Boolean matrix
\begin{equation}\label{eq:exact-hamming-boolean-matrix}
  K(D)=
  \begin{pmatrix}P(D)&Q\\I_d&I_d
  \end{pmatrix}
  \quad\text{satisfies}\quad
  \det K(D)=\sum_{ij}D_{ij}-t,\qquad
  \rank K(D)=d+\rank T_t(D)
\end{equation}
by~\eqref{eq:difference-block}.

Every input-dependent position of $K(D)$ copies a single bit of $D$;
all other positions are public constants. Therefore the local Boolean
matrices
\[
  \mathcal A(X)=K(X),\qquad
  \mathcal B(Y)=K(Y)\oplus K(0)
\]
satisfy $\mathcal A(X)\oplus\mathcal B(Y)=K(X\oplus Y)$:
the copies of private bits XOR, while Bob's public constants cancel.
The output has order $8m+2$, preserves corank, and is singular
precisely when $|x\oplus y|=t$. The target $t$ is represented by
the public Boolean block $H_t$, so no input entry exceeds one.

For sufficiently large $n$, take $m=\lfloor(n-2)/8\rfloor$ and
pad with an identity block. Equation~\eqref{eq:exact-hamming-quantum-bound}
then gives $Q^*_{1/3}(\SINGxor_{n,1})=\Omega(n^2)$, including the
adjacent-rank promise. This is a consequence of the cited symmetric-XOR
theorem and the local encoding. Boolean inputs are valid $k$-bit inputs
for every $k\ge1$, so the same lower bound holds for $\SINGxor_{n,k}$.

\subsection{A dense equation witness}

Fourier analysis of modular subset sums also appears in
Sherstov~\cite[Theorem~3.3]{sherstov-correction}, who bounds low-degree
Walsh coefficients for independently chosen weights. The witness below
requires a bound on every Walsh coefficient and coefficients of the
restricted form $a_i b_j$. We prove this bound using the product formula
for word bits and a smooth cutoff around exceptional frequencies.

\begin{theorem}[Dense equation witness]\label{thm:dense-equation-witness} For all sufficiently large integers $d,r$, put
  $N=2^r$ and $B=N^9-1$. There are public coefficients
  $1\le a_i,b_j\le N^4$, for $i,j\in[d]$, and a real signed measure
  $w$ on pairs

  \[
    X\in\{0,\ldots,N-1\}^{d\times d},\qquad
    Y\in\{0,\ldots,B\}^{2d\times2d}
  \]

  with the following properties. Let $F(X,Y)=+1$ if
  $a^{\mathsf T}Xb=\sum_{i,j}Y_{ij}$, and $F(X,Y)=-1$ otherwise.
  Encoding $X$ with $r$ bits per entry and $Y$ with $9r$ bits per entry
  identifies their domain with $G=\{0,1\}^{37d^2r}$. Then

  \begin{equation}
    \sum_z w(z)=0,\qquad
    \sum_z|w(z)|\le\frac{101}{100},\qquad
    \sum_zF(z)w(z)\ge\frac{99}{100},
    \label{eq:dense-witness-mass-correlation}
  \end{equation}

  and, for an absolute constant $c>0$,

  \begin{equation}
    \max_{\omega\in G}
    \left|\sum_{z\in G}w(z)(-1)^{\langle\omega,z\rangle}\right|
    \le2^{-cd^2r}.
    \label{eq:dense-witness-walsh-bound}
  \end{equation}

\end{theorem}

The proof below permits $c=1/2{,}000{,}000$.

\subsubsection{A signed source with an integer output}

Set

\[
  A=N^4,\qquad D=N/2,\qquad
  Q=2^{\lceil\log_2(d^2)\rceil}N^9.
\]

Thus $d^2N^9\le Q<2d^2N^9$. For independent uniform entries of
$X\in\{0,\ldots,N-1\}^{d\times d}$, write
$s(X)=\sum_{i,j}a_ib_jX_{ij}$ and start with

\begin{equation}
  \nu=\frac12\bigl(\mathcal L(X,s(X))-\mathcal L(X,s(X)+D)\bigr).
  \label{eq:dense-equation-signed-source}
\end{equation}

Both outputs lie in $[0,Q)$, because
$s(X)+D\le d^2N^8(N-1)+D<d^2N^9$.
The two supports are disjoint, so $\nu$ has mass zero and variation
one. It has correlation one with the sign of the integer equation
$s(X)=y$. We temporarily regard $y$ as an element of the cyclic
group $\mathbb Z/Q\mathbb Z$, but the two equations defining the
source are ordinary integer equations.

We will subtract a real correction $\mu$, with
$\|\mu\|_1\le1/100$, before splitting $y$ into bounded words.

\subsubsection{Frequency estimates for binary words}

For a mask $\beta$ on an $s$-bit word, let

\[
  \phi_{\beta,s}(\theta)
  =2^{-s}\sum_{x=0}^{2^s-1}(-1)^{\langle\beta,x\rangle}
  e^{2\pi i\theta x},\qquad
  m_s(\theta)=\prod_{j=0}^{s-1}
  \max\{|\cos(\pi2^j\theta)|,|\sin(\pi2^j\theta)|\}.
\]

Factoring over the bits gives a bound uniform in the mask:

\begin{equation}
  |\phi_{\beta,s}(\theta)|\le m_s(\theta)\le1.
  \label{eq:word-walsh-envelope}
\end{equation}

The factors have Lipschitz constants at most $\pi2^j$, so

\begin{equation}
  |m_s(u)-m_s(v)|\le\pi2^s|u-v|.
  \label{eq:word-envelope-lipschitz}
\end{equation}

Fix

\begin{equation}
  \delta=\frac1{10000},\quad \sigma=\frac1{50},\quad
  \eta=\frac1{1000},\quad \lambda_1=\frac1{40},\quad
  \lambda=\frac1{10}.
  \label{eq:dense-witness-exponents}
\end{equation}

All subsequent estimates hold for sufficiently large $r$.

\begin{lemma}[Large values occupy few short intervals]\label{lem:word-envelope-large-values}

  The set $S=\{\theta:m_r(\theta)>N^{-\delta}\}$ on the circle
  is covered by $J\le N^\sigma$ intervals of length at most $1/N$.

\end{lemma}
\begin{proof} Write $\theta=0.b_1b_2\cdots$ in binary. Whenever
  $b_{j+2}\ne b_{j+3}$, the $j$th factor of $m_r$ is at most
  $\cos(\pi/8)<2^{-1/16}$. Thus membership in $S$ permits fewer
  than $16\delta r$ transitions among $b_2,\ldots,b_{r+2}$.
  The number of possible prefixes is at most
  $4\sum_{i\le16\delta r}\binom ri$, each specifying an interval
  of length $2^{-r-2}$. Closed intervals cover the endpoints.

  For binary entropy $H$, the elementary binomial bound is
  $\sum_{i\le pr}\binom ri\le2^{H(p)r}$. At $p=16\delta=1/625$,
  $H(p)\le p\log_2(e/p)<11/625<\sigma$.
  This absorbs the factor $4$ for large $r$.
\end{proof}

\begin{lemma}[Counting a progression through short intervals]\label{lem:progression-interval-count}

  Suppose a set $S_0$ is covered by $J_0$ intervals of length at
  most $w_0$. Let $p/q$ be reduced, $1\le q\le A$, and
  $\varepsilon=\theta-p/q$ satisfy $|\varepsilon|\le1/(Aq)$.
  If $p_0$ is the fraction of $a\in\{1,\ldots,A\}$ with
  $a\theta\bmod1\in S_0$, then

  \begin{equation}
    p_0\le2J_0w_0+\frac{8J_0}{q}.
    \label{eq:progression-denominator-bound}
  \end{equation}

  If also $qw_0<1$ and $\varepsilon\ne0$, then

  \begin{equation}
    p_0\le\frac{6J_0}{A}
    +\frac{3J_0w_0}{Aq|\varepsilon|}.
    \label{eq:progression-approximation-bound}
  \end{equation}

\end{lemma}
\begin{proof} Break $1,\ldots,A$ into at most $2A/q$ blocks of $q$
  consecutive integers. In each block, $a\theta$ is a translate of
  the equally spaced $q$-point grid, with each point perturbed by
  at most $q|\varepsilon|\le1/q$. A covering interval contains at
  most $qw_0+4$ points per block, proving \eqref{eq:progression-denominator-bound}.

  For \eqref{eq:progression-approximation-bound}, fix a residue class of $a$ modulo $q$. Its points lie
  on a monotone progression of step $q\varepsilon$, in a segment
  of length at most $A|\varepsilon|\le1/q$. These containing
  segments start at the $q$ equally spaced grid points and all
  have the same direction. An interval of length $w_0<1/q$ meets
  at most three segments. Each intersection contains at most
  $2+w_0/(q|\varepsilon|)$ points, including points that wrap around
  the circle. Summing over the $J_0$ covering intervals and dividing by
  $A$ gives $p_0\le(3J_0/A)(2+w_0/(q|\varepsilon|))$, proving
  \eqref{eq:progression-approximation-bound}.
\end{proof}

\begin{lemma}[Dense products force a nearby rational]\label{lem:dense-products-rational-approximation}

  If independent uniform $a,b\in\{1,\ldots,A\}$ satisfy

  \[
    \Pr[m_r(ab\theta)>N^{-\delta}]>N^{-\eta},
  \]

  then some reduced fraction $p/q$ satisfies

  \begin{equation}
    1\le q\le N^\lambda,\qquad
    |\theta-p/q|\le \frac{N^{-1+\lambda}}{A^2q}.
    \label{eq:dense-products-rational-approximation}
  \end{equation}

\end{lemma}
\begin{proof} First consider a frequency $\alpha$ for which
  $\Pr_a[a\alpha\in S]>N^{-2\eta}$. Dirichlet's pigeonhole
  argument supplies a reduced $p/q$, $q\le A$, within $1/(Aq)$.
  Apply Lemma~\ref{lem:progression-interval-count} with $J_0\le N^\sigma$ and $w_0=1/N$.
  The first term of \eqref{eq:progression-denominator-bound} is at most $N^{-2\eta}/2$, so

  \[
    q\le16N^{\sigma+2\eta}\le N^{\lambda_1}.
  \]

  Now $qw_0<1$. The first term of \eqref{eq:progression-approximation-bound} is also at most
  $N^{-2\eta}/2$, giving

  \begin{equation}
    |\alpha-p/q|\le
    \frac{6N^{\sigma+2\eta}}{ANq}
    \le\frac{N^{-1+\lambda_1}}{Aq}.
    \label{eq:single-progression-rational-approximation}
  \end{equation}

  The case $\alpha=p/q$ already satisfies this inequality.

  Write $p_b=\Pr_a[m_r(ab\theta)>N^{-\delta}]$.
  The hypothesis implies that the fraction of $b$ with
  $p_b>N^{-2\eta}$ is greater than
  $N^{-\eta}-N^{-2\eta}\ge N^{-2\eta}$.
  For each such $b$, \eqref{eq:single-progression-rational-approximation} places $b\theta$ in the union of the
  rational neighborhoods with denominators at most $N^{\lambda_1}$.
  This union has a cover with

  \[
    J_0\le2N^{2\lambda_1},\qquad
    w_0=\frac{2N^{-1+\lambda_1}}A.
  \]

  Apply Lemma~\ref{lem:progression-interval-count} a second time, now to the progression $b\theta$.
  Again the first terms of both estimates are at most $N^{-2\eta}/2$.
  The first estimate gives
  $q\le32N^{2\lambda_1+2\eta}\le N^\lambda$ and hence $qw_0<1$.
  If $\theta=p/q$, the conclusion is immediate. Otherwise the second
  estimate gives

  \[
    |\theta-p/q|\le
    \frac{24N^{-1+3\lambda_1+2\eta}}{A^2q}
    \le\frac{N^{-1+\lambda}}{A^2q}.
  \]

  All constant factors are absorbed by the strict exponent margins
  in \eqref{eq:dense-witness-exponents}.
\end{proof}

\subsubsection{Choosing the row and column coefficients}

Set $L=\lfloor(1-4\lambda)r\rfloor$ and $q_0=2^{20}$.

\begin{lemma}[Few residues at odd denominators]\label{lem:odd-denominator-residues}

  For odd $3\le q\le N^\lambda$, define

  \[
    B_q=\{u\bmod q:m_L(u/q)\ge\tfrac12N^{-\delta}\}.
  \]

  Then $|B_q|\le8q^{1/5}$. If $q\le q_0$, then $B_q=\{0\}$.

\end{lemma}
\begin{proof} A qualifying residue has at most
  $t=16(\delta r+1)\le17\delta r$ transitions among its
  $L+1$ bits $b_2,\ldots,b_{L+2}$, by the proof of Lemma~\ref{lem:word-envelope-large-values}.
  Put $\ell=\lceil\log_2q\rceil+1$.
  There are $W=L+2-\ell\ge(1-5\lambda)r-1$ windows of length
  $\ell$ in these bits. Each transition occurs in at most $\ell-1$
  windows. At least half the windows therefore have at most
  $\lfloor2t(\ell-1)/W\rfloor$ transitions. Our constants give
  $2t/W<1/50$.

  For a fixed window position, multiplication by a power of two
  permutes residues modulo odd $q$. An $\ell$-bit prefix contains
  at most one point $u/q$, since $2^{-\ell}<1/q$.
  Counting pairs of a qualifying residue and one of its good windows
  gives

  \begin{equation}
    |B_q|\le4\sum_{i\le(\ell-1)/50}\binom{\ell-1}{i}
    \le4\,2^{H(1/50)(\ell-1)}
    \le8q^{1/5},
    \label{eq:odd-denominator-residue-count}
  \end{equation}

  where $H(1/50)\le(1/50)\log_2(50e)<1/5$.

  Every doubled nonzero residue modulo an odd $q$ stays nonzero,
  so its fractional part stays in $[1/q,1-1/q]$.
  Its binary expansion consequently has no constant run of length
  $\ell$. For $q\le q_0$, we have $\ell\le21$, forcing at least
  $(L+1)/21-1>17\delta r$ transitions. Thus only zero can qualify
  in this case.
\end{proof}

For uniform $a\in\{1,\ldots,A\}$ and $\gcd(p,q)=1$, it follows that

\begin{equation}
  \Pr[ap\bmod q\in B_q]\le
  \begin{cases}
    1/q,&3\le q\le q_0,\ q\text{ odd},\\
    q^{-1/2},&q_0<q\le N^\lambda,\ q\text{ odd}.
  \end{cases}
  \label{eq:odd-denominator-residue-probability}
\end{equation}

Indeed, the first probability is $\lfloor A/q\rfloor/A$.
The second is at most
$8q^{1/5}(1/q+1/A)\le16q^{-4/5}\le q^{-1/2}$.

Choose $a_1,\ldots,a_d,b_1,\ldots,b_d$ independently and uniformly
from $\{1,\ldots,A\}$. There is a choice satisfying both properties:

\begin{enumerate}
  \item For each grid frequency $\theta=t/Q$ whose probability in
    Lemma~\ref{lem:dense-products-rational-approximation} is at most $N^{-\eta}$, fewer than $99d^2/100$ pairs
    $(i,j)$ have $m_r(a_ib_j\theta)>N^{-\delta}$.
  \item For every reduced $p/q$ with odd $3\le q\le N^\lambda$, fewer
    than $9d/10$ entries of each coefficient vector belong to
    $\{a:ap\bmod q\in B_q\}$.

\end{enumerate}

For property 1, decompose the complete bipartite grid into $d$
perfect matchings, for example $j=i+s\bmod d$.
Within one matching the pairs $(a_i,b_j)$ are independent.
If at least $99d^2/100$ pairs are bad, one matching has at least
$99d/100$ bad pairs. The union bound over matchings and frequencies
is therefore at most

\begin{equation}
  Qd\,2^dN^{-99\eta d/100}
  \le2d^3\,2^dN^{\,9-99\eta d/100}.
  \label{eq:coefficient-grid-failure-probability}
\end{equation}

This is less than $1/4$ uniformly for all sufficiently large $d,r$.
For example, once $d\ge10^6$, the exponent of $N$ is at most
$-0.00098d$, which dominates $d+3\log_2d+1$ for large $r$.

By \eqref{eq:odd-denominator-residue-probability}, the failure probability of property 2 for one vector
is at most

\begin{equation}
  q_0^2\,2^d3^{-9d/10}
  +2^d\sum_{q>q_0}q^{\,1-9d/20}
  \le2^{40-2d/5}+2^{40-8d}.
  \label{eq:coefficient-residue-failure-probability}
\end{equation}

This is less than $1/8$ for all sufficiently large $d$.
A union bound for the two vectors and \eqref{eq:coefficient-grid-failure-probability} proves existence.
Fix such vectors. They depend only on $d,r$, and their selection
requires no communication.

\subsubsection{Fourier decay outside dyadic neighborhoods}

Call a frequency exceptional if it lies within

\begin{equation}
  w_q=\frac{N^{-1+\lambda}}{A^2q}
  \label{eq:dyadic-neighborhood-width}
\end{equation}

of a reduced dyadic fraction $p/q$ with $q\le N^\lambda$.
Centers are taken once modulo one, including $0/1$.
For every nonexceptional grid frequency,

\begin{equation}
  \prod_{i,j}m_r(a_ib_j\theta)\le N^{-\delta d^2/100}.
  \label{eq:nonexceptional-product-decay}
\end{equation}

If the probability in Lemma~\ref{lem:dense-products-rational-approximation} is at most $N^{-\eta}$, property 1
gives \eqref{eq:nonexceptional-product-decay} immediately. Otherwise use \eqref{eq:dense-products-rational-approximation} and write
$q=2^sv$ with $v$ odd. Nonexceptionality forces $v\ge3$.

Put $s_0=\lfloor(1-2\lambda)r\rfloor$. For $c\le A^2$, \eqref{eq:word-envelope-lipschitz} gives

\begin{equation}
  |m_{s_0}(c\theta)-m_{s_0}(cp/q)|
  \le\pi2^{s_0}A^2|\theta-p/q|
  \le\pi N^{-\lambda}\le\tfrac12N^{-\delta}.
  \label{eq:rational-perturbation-bound}
\end{equation}

Since $s\le\lambda r$ and $s_0-s\ge L$, dropping factors shows

\begin{equation}
  m_r(c\theta)>N^{-\delta}
  \quad\Longrightarrow\quad
  m_L(cp/v)>\tfrac12N^{-\delta}.
  \label{eq:odd-denominator-reduction}
\end{equation}

Property 2 for the vector $b$, applied to denominator $v$,
implies that at least $d/10$ indices $j$ have $v\nmid b_j$,
because $0\in B_v$. For any such $j$, reduce
$b_jp/v=p'/v'$; then $v'\ge3$ is odd.
Property 2 for $a$, applied to this reduced fraction, gives at
least $d/10$ indices $i$ with
$m_L(a_ib_jp/v)<N^{-\delta}/2$.
By \eqref{eq:odd-denominator-reduction}, at least $d^2/100$ pairs therefore have
$m_r(a_ib_j\theta)\le N^{-\delta}$, proving \eqref{eq:nonexceptional-product-decay}.

\subsubsection{Removing dyadic neighborhoods at small variation cost}

The shift $D$ is divisible by every exceptional denominator $q$.
Thus the Fourier transform of \eqref{eq:dense-equation-signed-source} vanishes at each exceptional
center. A smooth cutoff removes a neighborhood of each center
while changing the measure by only a small amount.

We use the de la Vallée Poussin combination of Fejér kernels;
see D\~ung, Temlyakov, and Ullrich~\cite[Section~2.1 of the author manuscript]{dung-temlyakov-ullrich}.
We sample the kernel on $\mathbb Z/Q$ and prove the discrete shift
estimate needed here. For $1\le M$ with $4M<Q$, define

\[
  F_M(y)=\frac1{QM}
  \left|\sum_{j=0}^{M-1}e^{2\pi i jy/Q}\right|^2,
  \qquad V_M=2F_{2M}-F_M.
\]

The kernel $F_M$ is nonnegative and sums to one. The cyclic Fourier
multiplier of $V_M$ is the trapezoid

\begin{equation}
  v_M(t)=
  \begin{cases}
    1,&|t|\le M,\\
    2-|t|/M,&M\le|t|\le2M,\\
    0,&|t|\ge2M,
  \end{cases}
  \label{eq:trapezoidal-fourier-multiplier}
\end{equation}

where $|t|$ denotes the distance to zero modulo $Q$.
In particular, $0\le v_M\le1$.

\begin{lemma}[Small changes under a shift]\label{lem:cutoff-kernel-shift-bound}

  With $\tau_D f(y)=f(y-D)$,

  \begin{equation}
    \|V_M-\tau_DV_M\|_1\le20\pi DM/Q.
    \label{eq:cutoff-kernel-shift-bound}
  \end{equation}

\end{lemma}
\begin{proof} Let $S_M(y)=\sum_{j<M}e^{2\pi i jy/Q}$.
  Orthogonality gives $\|S_M\|_2^2=QM$, and

  \[
    \|S_M-\tau_DS_M\|_2^2
    =Q\sum_{j<M}|1-e^{-2\pi i jD/Q}|^2
    \le4\pi^2D^2M^3/Q.
  \]

  The inequality
  $\||S|^2-|T|^2\|_1\le\|S-T\|_2(\|S\|_2+\|T\|_2)$
  therefore gives
  $\|F_M-\tau_DF_M\|_1\le4\pi DM/Q$.
  Applying this to $2F_{2M}-F_M$ proves \eqref{eq:cutoff-kernel-shift-bound}.
\end{proof}

For every reduced dyadic center $p/q$ in \eqref{eq:dyadic-neighborhood-width}, take
$M_q=\lceil Qw_q\rceil$, and modulate $V_{M_q}$ to that center.
Let $K$ be the sum of these modulated kernels and $f(t)$ its
Fourier multiplier. Each center is a grid point because $q\mid Q$.
For large $r$, the trapezoid supports are disjoint: distinct centers
are separated by at least $N^{-\lambda}$, whereas their half-widths
are at most $2w_1+2/Q$.
Consequently,

\begin{equation}
  0\le f(t)\le1,\qquad
  f(t)=1\quad\text{at every exceptional grid frequency}.
  \label{eq:dyadic-cutoff-multiplier}
\end{equation}

The centers occur in opposite pairs with equal widths, so $K$ is
real and even, and $f(t)=f(-t)$. Define $\mu=K*_y\nu$,
convolving only in the output coordinate,
and put $\nu'=\nu-\mu$.

At a center $p/q$, modulation does not change the shift estimate:
$e^{2\pi i pD/q}=1$. For each fixed $X$, convolve the signed two-atom
measure in~\eqref{eq:dense-equation-signed-source} with that kernel. The result has variation at
most half the bound in~\eqref{eq:cutoff-kernel-shift-bound}.
Averaging over $X$ and summing over centers yields

\begin{equation}
  \begin{aligned}
    \|\mu\|_1
    &\le10\pi D\sum_{p/q}\frac{M_q}{Q}\\
    &\le10\pi D\left(
      (1+\lambda r)\frac{N^{-1+\lambda}}{A^2}
    +\frac{N^\lambda}{Q}\right)\\
    &\le5\pi(\lambda r+2)N^{-8+\lambda}
    \le\frac1{100}.
  \end{aligned}
  \label{eq:dyadic-correction-variation}
\end{equation}

For the second line, each dyadic denominator has at most $q$
reduced numerators, and there are at most $1+\lambda r$
denominators. The total number of centers is at most $N^\lambda$.
The last bound is uniform in $d$, which is essential when the
matrix dimension is much larger than the word range.

Both $\nu$ and $\mu$ have mass zero, even for each fixed $X$.
Thus $\nu'$ has mass zero, variation at most $101/100$, and
correlation at least $99/100$ with the sign of $s(X)=y$.
The correction may wrap around modulo $Q$; its entire variation
bounds the resulting loss of correlation with this integer equation.

\subsubsection{Encoding the output and bounding Walsh coefficients}

\begin{proof}[Completion of Theorem~\ref{thm:dense-equation-witness}]

  List the $4d^2$ positions of a $2d\times2d$ matrix in a fixed
  public order. For $0\le y<Q$, define

  \begin{equation}
    \operatorname{enc}(y)_\ell
    =\min\{B,\max\{0,y-(\ell-1)B\}\},
    \qquad 1\le\ell\le4d^2.
    \label{eq:bounded-word-output-encoding}
  \end{equation}

  This fills entries with $B$, followed by at most one intermediate
  entry and then zeros. Since $4d^2B\ge Q-1$, it satisfies

  \begin{equation}
    0\le\operatorname{enc}(y)_\ell\le B,\qquad
    \sum_\ell\operatorname{enc}(y)_\ell=y.
    \label{eq:output-encoding-sum}
  \end{equation}

  Push $\nu'$ forward under $(X,y)\mapsto(X,\operatorname{enc}(y))$
  to obtain $w$. This preserves mass, cannot increase variation,
  and preserves the equation's label. Hence \eqref{eq:dense-witness-mass-correlation} follows.
  This encoding defines the witness distribution; the communication
  reduction simply copies the resulting word variables.

  It remains to check every Walsh coefficient of $w$, including masks
  on all of the bounded output words. Fix masks $\beta_{ij}$ on $X$
  and an arbitrary mask $\gamma$ on the encoded output. The function
  $\chi_\gamma(\operatorname{enc}(y))$ is a sign function on
  $\mathbb Z/Q\mathbb Z$. Its cyclic Fourier expansion has the form

  \[
    \chi_\gamma(\operatorname{enc}(y))
    =\sum_{t=0}^{Q-1}c_t e^{2\pi i ty/Q},
    \qquad \sum_t|c_t|^2=1.
  \]

  The corresponding Walsh coefficient of $w$ is exactly

  \begin{equation}
    \sum_{t=0}^{Q-1}c_t(1-f(t))
    \frac{1-e^{2\pi iDt/Q}}2
    \prod_{i,j}\phi_{\beta_{ij},r}(a_ib_jt/Q).
    \label{eq:encoded-witness-walsh-coefficient}
  \end{equation}

  Every exceptional frequency is removed by \eqref{eq:dyadic-cutoff-multiplier}. At every other
  frequency, \eqref{eq:word-walsh-envelope} and \eqref{eq:nonexceptional-product-decay} bound the product by
  $N^{-\delta d^2/100}$; the other two factors have absolute value
  at most one. Cauchy--Schwarz therefore bounds \eqref{eq:encoded-witness-walsh-coefficient} by

  \begin{equation}
    \sqrt Q\,N^{-\delta d^2/100}
    \le\sqrt2\,d\,N^{9/2-\delta d^2/100}
    \le N^{-\delta d^2/200}.
    \label{eq:encoded-witness-walsh-bound}
  \end{equation}

  The last inequality holds uniformly for all sufficiently large
  $d,r$. With $c=\delta/200$, this is \eqref{eq:dense-witness-walsh-bound}, completing the proof
  of Theorem~\ref{thm:dense-equation-witness}.
\end{proof}

\subsection{Reduction from the dense equation to XOR singularity}

Use the public coefficients of Theorem~\ref{thm:dense-equation-witness} and
$Z=\diag(X,Y)$, $\alpha=(a,-\mathbf1_{2d})$,
$\beta=(b,\mathbf1_{2d})$ in~\eqref{eq:weighted-sum-matrix}.
The resulting matrix $T=T_{\alpha,\beta}(Z)$ has order $6d+1$ and
\[
  \det T=a^TXb-\sum_{i,j}Y_{ij},\qquad
  \corank T=\mathbf1[a^TXb=\sum_{i,j}Y_{ij}].
\]
Its only negative entries are public $-1$s. Write $T=U-V$ using
its positive and negative parts and form
\begin{equation}\label{eq:dense-equation-singularity-matrix}
  E=
  \begin{pmatrix}U&V\\I_{6d+1}&I_{6d+1}
  \end{pmatrix}.
\end{equation}
By~\eqref{eq:difference-block}, $E$ has the same determinant and
corank as $T$. It has order $12d+2$, and every entry fits in $9r$
bits, since the public coefficients are at most $2^{4r}<2^{9r}$.
Each private entry is one word from $X$ or $Y$. Given XOR shares of
these words, each party simply copies its share into the appropriate
positions; Alice supplies public constants and Bob supplies zeros.
Their output XOR is $E$. Every input satisfies the adjacent-rank promise.

\begin{proof}[Completion of Theorem~\ref{thm:xor-lower-bound}]
  For $G=\{0,1\}^{37d^2r}$ let $F$ and $w$ be as in
  Theorem~\ref{thm:dense-equation-witness}, and put
  $W_{u,v}=|G|^{-1}w(u\oplus v)$. Its mass, variation, and correlation
  with $F(u\oplus v)$ equal those of $w$. Walsh characters diagonalize
  XOR convolution, so
  \[
    \|W\|=|G|^{-1}\max_{\omega\in G}|\widehat w(\omega)|
    \le |G|^{-1}2^{-cd^2r}.
  \]
  Here the transform of $w$ sums weights against characters without
  normalization. Uniform reference probabilities give $T_W=|G|W$.
  For the acceptance matrix of a $q$-qubit protocol, Lemma~\ref{lem:quantum-witness-pairing}
  yields
  \[
    \frac18\le\frac{99}{200}-\frac{101}{300}
    \le\langle p,W\rangle\le2^q|G|\|W\|\le2^{q-cd^2r}.
  \]
  Thus $q\ge cd^2r-3$. Compose with~\eqref{eq:dense-equation-singularity-matrix} and take
  $d=\lfloor(n-2)/12\rfloor$, $r=\lfloor k/9\rfloor$.
  For $n,k$ above independent absolute thresholds these are
  $\Theta(n),\Theta(k)$; identity padding gives the claimed lower bound.
  For the remaining bounded set of bit lengths, the Boolean lower
  bound from Section~\ref{sec:boolean-xor-lower-bound} is
  $\Omega(n^2)=\Omega(n^2k)$ with a smaller absolute constant.
  Every reduction preserves corank, so the proof also applies to the
  promised problem.
\end{proof}

\subsection{Boolean entrywise-product singularity}\label{sec:boolean-product-singularity}

The Boolean baseline from Section~\ref{sec:boolean-xor-lower-bound} also determines
entrywise-product singularity for Boolean inputs, where multiplication
coincides with bitwise AND. Thus the following corollary does not rely on
Theorem~\ref{thm:xor-lower-bound}; the Boolean XOR bound suffices. The relevant reductions are deterministic, local, and
preserve corank with constant dimension growth. They show that the
corollary remains true with a natural border restriction.

\begin{corollary}[Boolean entrywise products]\label{cor:boolean-product-singularity}
  For $A,B\in\{0,1\}^{n\times n}$, deciding whether $A\bitand B$ is
  singular over $\mathbb Q$ has randomized and entangled quantum
  complexity $\Theta(n^2)$, where $\bitand$ is the entrywise product
  (or entrywise AND) symbol. The lower bounds hold even when the product
  has rank $n-1$ or $n$ and both inputs have all-one first rows and columns.
\end{corollary}

For Boolean matrices $X,Y$ of the same shape, let $J$ denote the all-one
matrix of that shape. The identity
$X\oplus Y=(X\bitand(J-Y))+((J-X)\bitand Y)$ separates XOR into
two Boolean products. A block matrix computes their sum without
increasing their alphabet. A second block construction attaches the
all-one border while adding two pivots and preserving corank. Thus
this result needs no separate hard distribution for product inputs.

\begin{proof}[Proof of Corollary~\ref{cor:boolean-product-singularity}]
  Given Boolean $m\times m$ XOR inputs $X,Y$, put
  $P=X\bitand(J-Y)$ and $Q=(J-X)\bitand Y$. Then $P+Q=X\oplus Y$.
  By~\eqref{eq:four-block},
  \[
    M=
    \begin{pmatrix}P&0&I\\0&Q&I\\I&I&0
    \end{pmatrix}
    \quad\text{has rank}\quad 2m+\rank(X\oplus Y).
  \]
  It satisfies $M=A_X\bitand B_Y$ for these locally constructed Boolean matrices:
  \begin{equation}\label{eq:boolean-product-shares}
    A_X=
    \begin{pmatrix}X&0&I\\0&J-X&I\\I&I&0
    \end{pmatrix},\qquad
    B_Y=
    \begin{pmatrix}J-Y&J&J\\J&Y&J\\J&J&J
    \end{pmatrix}.
  \end{equation}
  For a Boolean matrix $D$ of order $q$, define
  \begin{equation}\label{eq:boolean-border}
    N(D)=
    \begin{pmatrix}
      1&\mathbf1_q^{\mathsf T}&1\\\mathbf1_q&D&0\\1&0^{\mathsf T}&0
    \end{pmatrix}.
  \end{equation}
  Eliminating the two pivots in its last row and column gives
  $\rank N(D)=2+\rank D$ and $\det N(D)=-\det D$.
  Moreover, $N(A)\bitand N(B)=N(A\bitand B)$, and both inputs have
  all-one first rows and columns. Take $m=\lfloor(n-2)/3\rfloor$,
  pad both factors in~\eqref{eq:boolean-product-shares} with the same
  identity block to order $n-2$, and apply $N$. Corank is preserved
  throughout. The adjacent-rank lower bound from
  Section~\ref{sec:boolean-xor-lower-bound}, as a consequence of
  Zhang--Shi~\cite[Theorem~2]{zhang-shi-xor}, gives
  $\Omega(m^2)=\Omega(n^2)$.
  Sending Alice's $n^2$ bits proves the upper bound.
\end{proof}

Together,~\eqref{eq:and-to-xor-rank} and~\eqref{eq:boolean-product-shares}
show that Boolean XOR and Boolean entrywise-product singularity
are equivalent under local corank-preserving reductions with
constant dimension growth.

\section*{AI disclosure}
The author used OpenAI Codex and Claude Code to assist with developing, checking, and
simplifying proof arguments; researching related work and verifying
citations; and drafting and revising the manuscript.
The author takes full responsibility for the correctness
of all content, including the proofs and references.

\clearpage
\bibliographystyle{alphaurl}
\bibliography{integer-matrix-singularity}
\end{document}